\documentclass{article}

\usepackage[utf8]{inputenc}
\usepackage{amssymb}
\usepackage{amsthm}
\usepackage{graphicx}
\usepackage{grffile}
\usepackage{booktabs}
\usepackage{subfigure}
\usepackage[all]{xy}
\newcommand{\wOnePointFive}{0.70\textwidth}
\newcommand{\wTwo}{0.47\textwidth}
\newcommand{\wThree}{0.30\textwidth}

\newtheorem{mydef}{Definition}
\newtheorem{theorem}{Theorem}

\newcommand{\wToy}{0.288\textwidth}
\newcommand{\raiseToy}{+2.1cm}

\newcommand{\node}{*=0{\bullet}}
\newcommand{\circlenode}[1]{*++[o][F-]{#1}}

\begin{document}

\title{
  Applications of Structural Balance \\ in Signed Social Networks  
}

\author{
  J\'{e}r\^{o}me Kunegis\\ 
  University of Koblenz--Landau, Germany\\
  \texttt{kunegis@uni-koblenz.de}
}

\maketitle

\begin{abstract}
  We present measures, models and link prediction algorithms based on
  the structural balance in signed social networks.  Certain social
  networks contain, in addition to the usual \emph{friend} links,
  \emph{enemy} links. These networks are called signed social
  networks. A classical and major concept for signed social networks is
  that of structural balance, i.e., the tendency of triangles to be
  \emph{balanced} towards including an even number of negative edges,
  such as friend-friend-friend and friend-enemy-enemy triangles.  In
  this article, we introduce several new signed network analysis methods
  that exploit structural balance for measuring partial balance, for
  finding communities of people based on balance, for drawing signed
  social networks, and for solving the problem of link prediction.
  Notably, the introduced methods are based on the signed graph
  Laplacian and on the concept of signed resistance distances.  We
  evaluate our methods on a collection of four signed social network
  datasets.
\end{abstract}

\section{Introduction}
Signed social networks are such social networks in which social ties can
have two signs:  friendship and enmity.  Signed social networks have
been studied in sociology and anthropology\footnote{See for instance
  Figure~\ref{fig:gama-nocluster}}, and are now found on certain
websites such as
Slashdot\footnote{slashdot.org} and
Epinions\footnote{www.epinions.com}.  
In addition to the usual social network analyses, the signed structure of these
networks allows a new range of studies to be performed, related to the
behavior of edge sign distributions within the graph. A major
observation in this regard is the now classical result of \emph{balance
  theory} by~\cite{b355}, stipulating that signed social
networks tend to be balanced in the sense that its nodes can be partitioned
into two sets, such that nodes within each set are connected only by
friendship ties, and nodes from different sets are only connected by
enmity ties.  This observation is not to be understood in an absolute
sense~-- in a large social network, a single wrongly signed edge would
render a network unbalanced. Instead, this is to be understood as a
tendency, which can be exploited to enhance the analytical and
predictive power of network analysis methods for a wide range of
applications. 

\begin{figure}
  \centering
  \includegraphics[width=\wOnePointFive]{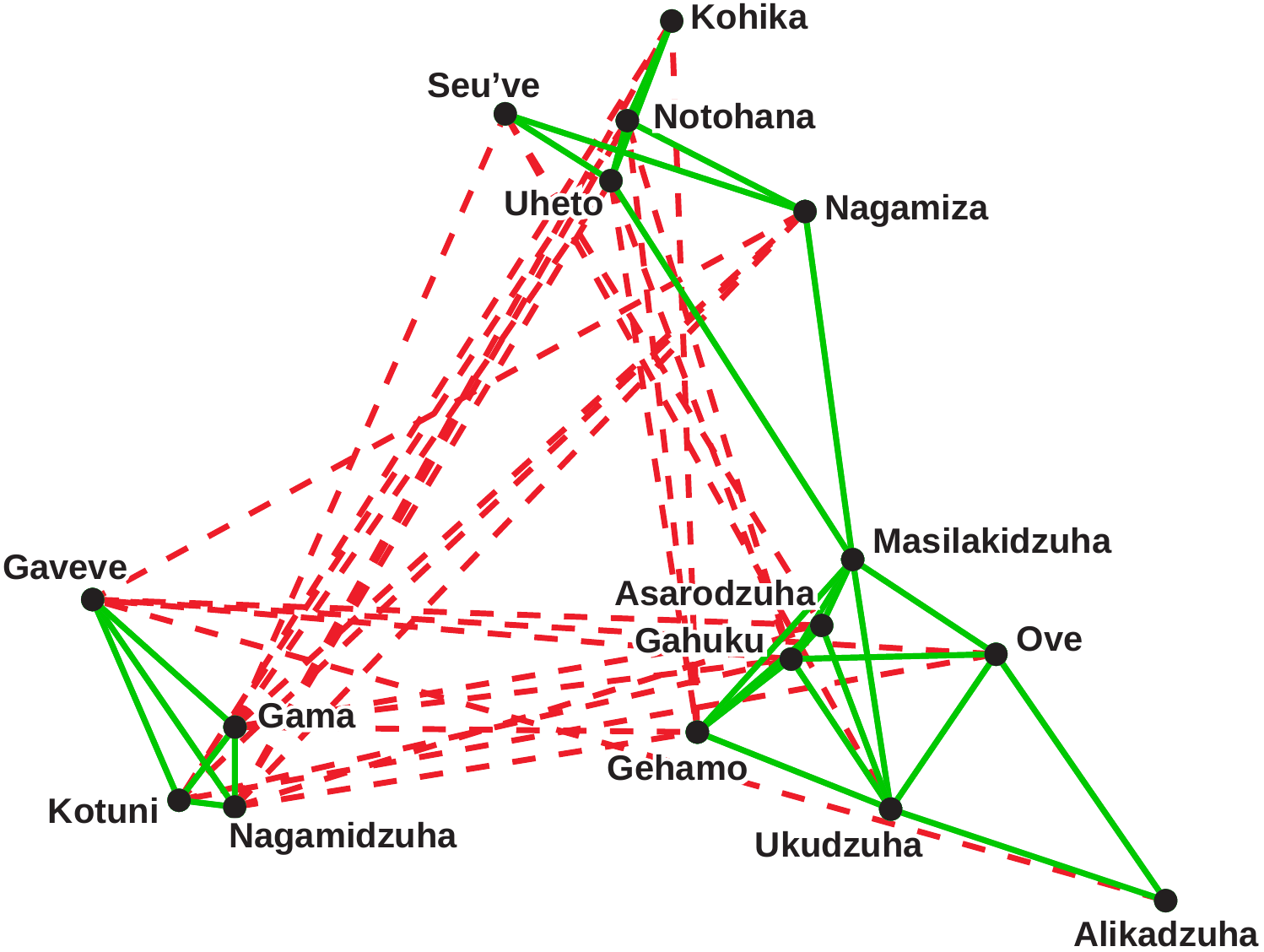} 
  \caption[*]{
    A small example of a signed social network from anthropology: 
    The tribal groups of the Eastern Central Highlands of New Guinea
    from the study of Read~\cite{b322}. 
    Individual tribes are the vertices of this network, with friendly
    relations shown as green edges and antagonistic relations shown as
    red edges.  
  }
  \label{fig:gama-nocluster}
\end{figure}

In this article, we present ways to measure and exploit structural
balance of signed social networks for graph drawing, measuring conflict,
detecting communities and predicting links. 
In particular, we introduce methods based on \emph{algebraic graph
  theory}, i.e., the representation of graphs by matrices. In ordinary
network analysis applications, algebraic graph theory has the advantage
that a large range of powerful algebraic methods become available to
analyse networks. In the case of signed networks, an 
additional advantage is that structural balance, which is inherently a
multiplicative construct as illustrated by the rule \emph{the enemy
  of my enemy is my friend}, maps in a natural way onto the algebraic
representation of networks as matrices. As we will see, this makes not
only signed network analysis methods seamlessly take into account
structural balance theory, it also simplifies calculation with matrices
and vectors, as the multiplication rule is build right into the
definition of their operations. 

In the rest of article, the individual methods are not presented in
order of possible applications, but in order of complexity, building on
each other. The breakdown is as follows:
\begin{itemize}
\item Section~\ref{sec:background} introduces the concept of a signed
  social network, gives necessary mathematical definitions and
  presents a set of four signed social networks that are used
  throughout the paper. 
\item Section~\ref{sec:clusco} defines structural balance and
  introduces a basic but novel measure for quantifying it: the signed
  clustering coefficient. 
\item Section~\ref{sec:drawing} reviews the problem of drawing signed
  graphs, and derives from it the signed Laplacian matrix which arises
  naturally in that context. 
\item Section~\ref{sec:laplacian} gives a proper mathematical
  definition of the signed Laplacian matrix, and proves its basic
  properties. 
\item Section~\ref{sec:conflict} introduces the notion of
  \emph{algebraic conflict}, a second way of quantifying structural
  balance, based on a spectral analysis of the signed Laplacian
  matrix. 
\item Section~\ref{sec:clustering} describes the signed graph
  clustering problem, and shows how its solution leads to another
  derivation of the signed Laplacian matrix. 
\item Section~\ref{sec:prediction} reviews the problem of link
  prediction in signed networks, and shows how it can be solved by the \emph{signed
    resistance distance}. 
\end{itemize}
Section~\ref{sec:conclusion} concludes the article. 
This article is partially based on material previously published by the author in
conference papers
\cite{kunegis:phd,kunegis:slashdot-zoo,kunegis:negative-resistance,kunegis:netflix-srd,kunegis:signed-kernels}.  

\section{Background:  Signed Social Networks}
\label{sec:background}
Negative edges can be found in many types of social networks, 
to model enmity in addition to friendship, distrust in addition
to trust, or positive and negative ratings between users. 
Early uses of signed social networks can be found
in anthropology, where negative edges have been used to denote
antagonistic relationships between tribes~\cite{b323}. 
In this context, the sociological notion of balance is defined as the
absence of negative cycles, i.e., the absence of cycles with an odd
number of negative edges~\cite{b284,b355}.  
Other cases of signed social networks include student
relationships~\cite{b493} and voting processes~\cite{b551}. 

Recent studies~\cite{b270}
describe the social network extracted from Essembly, an ideological
discussion site that allows users to mark other users as \emph{friends},
\emph{allies} and \emph{nemeses}, and discuss the semantics of the three
relation types.  These works model the different types of edges by means
of three subgraphs.  
Other recent work considers the task of
discovering communities from social networks with negative
edges~\cite{b233}.   

In trust networks, nodes represent persons or other entities, and links
represent trust relationships. 
To model distrust, negative edges are then used. 
Work in that field has mostly focused on defining global trust measures
using path lengths or adapting PageRank~\cite{b325,b236,b235,b237,b234}. 

In applications where users can rate each other, we can model ratings as
\emph{like} and \emph{dislike}, giving rise to positive and negative
edges, for instance on online dating sites~\cite{b311}.  

An example of a small signed social network is given by the tribal
groups of the Eastern Central Highlands of New Guinea from the study
of Read~\cite{b322} in Figure~\ref{fig:gama-nocluster}.    
This dataset describes the relations between sixteen tribal groups of
the Eastern Central Highlands of New Guinea~\cite{b323}.
Relations between tribal groups in the Gahuku--Gama alliance structure
can be friendly (\emph{rova}) or antagonistic (\emph{hina}).
In addition, four large signed social networks extracted from the Web will be used throughout the
article. All datasets are part of the Koblenz Network
Collection~\cite{konect}.  The datasets are summarized in
Table~\ref{tab:datasets}.

\begin{table}
  \caption{
    \label{tab:datasets}
    The signed social network datasets used in this article.  The first
    four datasets are large; the last one is small and serves as a
    running example. 
  }
\makebox[\textwidth]{
  \centering
%  \scalebox{0.80}{
    \begin{tabular}{l l r r r}
      \toprule
      \textbf{Network} & \textbf{Type} & \textbf{Vertices ($|V|$)} & \textbf{Edges ($|E|$)} & \textbf{Percent Negative} \\
      \midrule
      Slashdot Zoo \cite{kunegis:slashdot-zoo} 
      & Directed & 79,120 & 515,581 & 23.9\% \\ 
      Epinions \cite{b367} 
      & Directed & 131,828 & 841,372 & 14.7\% \\
      Wikipedia elections \cite{b551} 
      & Directed & 8,297 & 107,071 & 21.6\% \\
      Wikipedia conflicts \cite{konect:brandes09} 
      & Undirected & 118,100 & 2,985,790 & 19.5\% \\
      \midrule
      Highland tribes \cite{b322}
      & Undirected & 16 & 58 & 50.0\% \\
      \bottomrule
    \end{tabular}
%  }
    }
\end{table}

\paragraph{Definitions}
Mathematically, an undirected signed graph can be defined as $G = (V,E,\sigma)$,
where $V$ is the vertex set, $E$ is the edge set, and $\sigma: E
\rightarrow \{-1, +1\}$ is the sign function~\cite{b324}.  
The sign function $\sigma$ assigns a positive or negative sign to each
edge.  
The fact that two edges $u$ and $v$ are adjacent will be denoted by $u
\sim v$.  The degree of a node $u$ is defined as the number of its
neighbors, and can be written as 
\begin{eqnarray*}
  d(u) = \{ v \mid u \sim v \}. 
\end{eqnarray*}
A directed signed network will be noted as $G=(V,E,\sigma)$, in which
$E$ is the set of directed edges (or \emph{arcs}). 

\paragraph{Algebraic Graph Theory}
Algebraic graph theory is the branch of graph theory that represents
graphs using algebraic structures in order to exploit the powerful
methods of algebra in graph theory. The main tool of algebraic graph
theory is the representation of graphs as matrices, in particular the
adjacency matrix and the Laplacian matrix. 
In the following, all matrices are real. 

Given a signed graph $G=(V,E,\sigma)$, its adjacency matrix $\mathbf A \in \mathbb
R^{|V|\times |V|}$ is defined as
\begin{eqnarray*}
  \mathbf A_{uv} &=& \left\{ \begin{array}{ll} 
    \sigma(\{u,v\}) & \mathrm{when } \{u,v\} \in E \\ 
    0 & \mathrm{when } \{u,v\} \notin E 
  \end{array} \right. 
\end{eqnarray*}
The adjacency matrix is square and symmetric. 

%% The adjacency matrix $\mathbf A\in \mathbb R^{|V|\times |V|}$ of a signed graph is defined
%% using $\mathbf A_{uv} = \sigma(\{u,v\})$ when $\{u,v\}\in E$ and $\mathbf A_{uv} =
%% 0$ otherwise.  

The diagonal degree matrix $\mathbf D$ of a signed graph is defined
using $\mathbf D_{uu} = d(u)$. Note that the degrees, and thus the matrix
$\mathbf D$, is independent of the sign function~$\sigma$. 

The assumption of structural balance lends itself to using algebraic
methods based on the adjacency matrix of a signed network.
To see why this is true, consider that the square $\mathbf A^2$ 
contains at its entry $(u,v)$ a sum of paths of length
two between $u$ and $v$ weighted positively or negatively depending on
whether a third positive edge between $u$ and $v$ would lead to a
balanced or unbalanced triangle. 

Finally, the Laplacian matrix of any graph is defined as $\mathbf L =
\mathbf D - \mathbf A$. It is this matrix $\mathbf L$ that will play a central
role for graph drawing, graph clustering and link prediction. 

\section{Measuring Structural Balance: \\ The Signed Clustering Coefficient}
\label{sec:clusco}
In a signed social network, the relationship between two connected nodes
can be positive or negative. When looking at groups of three persons,
four combinations of positive and negative edges are possible (up to
permutations), some being more 
likely than others. An observation made in actual social groups is that
triangles of positive and negative edges tend to be balanced. For
instance, a triangle of three positive edges is balanced, as is a
triangle of one positive and two negative edges.  On the other hand, a
triangle of two positive and one negative edge is not balanced.  The
case of three negative edges can be considered balanced, when
considering the three persons as three different groups, or unbalanced,
when allowing only two groups. 

This characterization of balance can be generalized to the complete
signed network, resulting in the following definition:
\begin{mydef}[Harary, 1953]
% \cite{b355} Using this as the brackets-argument would put
% double-parentheses around it. 
  \label{def:balance}
  A connected signed graph is balanced when its
  vertices can be partitioned into two groups such that all positive edges
  connect vertices within the same group, and all negative edges connect
  vertices of the two different groups.
\end{mydef}
Figure~\ref{fig:balance} shows a balanced graph partitioned into two
vertex sets.  
The concept of structural balance can also be illustrated with the
phrase \emph{the enemy of my enemy is my friend} and its permutations. 

Equivalently, unbalanced graphs can be defined as those
graphs containing a cycle with an odd number of negative edges,
as shown in Figure~\ref{fig:conflict}. 
To prove that the balanced graphs are exactly those that do not
contain cycles with an odd number of edges, consider that any cycle in a
balanced graph has to cross sides an even number of times.  On
the other hand, any balanced graph can be partitioned into two
vertex sets by depth-first traversal while assigning each vertex to 
a partition such that the balance property is fulfilled.  Any
inconsistency that arises during such a labeling leads to a cycle with
an odd number of negative edges. 

% Code for two figures side-by-side taken from
% http://texblog.org/2007/08/01/placing-figurestables-side-by-side-minipage/
\begin{figure}
  \begin{minipage}[b]{0.45\linewidth}
  \centering
  \includegraphics[width=0.9\textwidth]{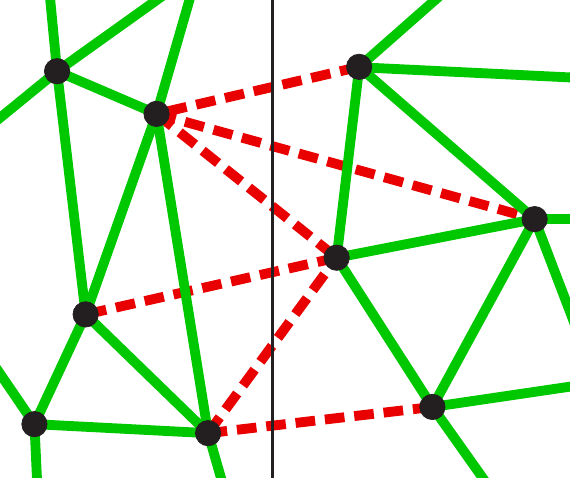}
  \caption{
    The nodes of a graph without negative cycles
    can be partitioned into two sets such that all edges
    inside of each group are positive and all edges between the two
    groups are negative. 
    We call such a graph balanced. 
  }
  \label{fig:balance}
  \end{minipage}
  \hspace{0.5cm}
  \begin{minipage}[b]{0.45\linewidth}
  \centering
  \includegraphics[width=0.9\textwidth]{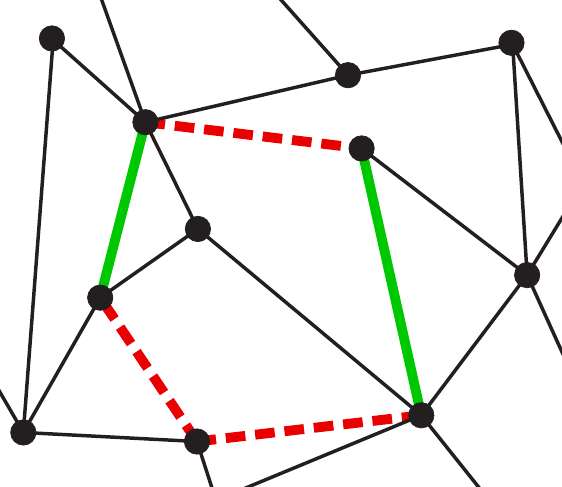}
  \caption{
    An unbalanced graph contains at least one
    cycle with an odd number of negative edges.  
    Such a graph cannot be partitioned into two sets with all negative
    edges across the sets and positive edges within the sets. 
    %% Negative edges are
    %% shown as dashed red lines, positively weighted 
    %% edges are shown 
    %% as solid green lines, and edges that are not part of the cycle in black. 
  }
  \label{fig:conflict}
  \end{minipage}
\end{figure}

In large signed social networks such as those given in
Table~\ref{tab:datasets}, it cannot be expected that the full network is
balanced, since already a single unbalanced triangle makes the full
network unbalanced. Instead, we need a measure of balance that
characterizes to what extent a signed network is balanced.  To that end,
we extend a well-establish measure in network analysis, the clustering
coefficient, to signed networks, giving the signed clustering
coefficient.  We also introduce the relative signed clustering
coefficient and give the values observed in our example datasets.

The clustering coefficient is a characteristic number of a graph taking
values between zero and one, denoting the tendency of the graph nodes
to form small clusters.  
The clustering coefficient was introduced in~\cite{b228} and an
extension for positively weighted edges 
proposed in~\cite{b269}. 
The signed clustering coefficient we 
define denotes the tendency of small clusters to be \emph{balanced}, and
takes on values between
$-1$ and $+1$.  The relative signed clustering coefficient will be
defined as the quotient between the two. 

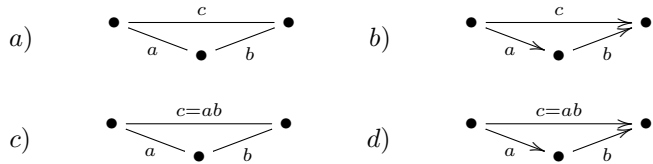
\begin{figure}
  \centerline{\xymatrix @R-31pt @C-2pt {
      & {\bullet} \ar@{-}[rr]^{c} \ar@{-}[rdd]_{a} & & {\bullet} &      & {\bullet} \ar@{->}[rr]^{c} \ar@{->}[rdd]_{a} & & {\bullet} \\ 
      {a)} & & & &   {b)}  \\
      & & {\bullet} \ar@{-}[ruu]_{b} & &     & & {\bullet} \ar@{->}[ruu]_{b} & 
  }}
  \vspace{0.3cm}
  \centerline{\xymatrix @R-31pt @C-2pt {
      & {\bullet} \ar@{-}[rr]^{c = ab} \ar@{-}[rdd]_{a} & & {\bullet} &      & {\bullet} \ar@{->}[rr]^{c = ab} \ar@{->}[rdd]_{a} & & {\bullet} \\
      {c)} & & & &    {d)} \\
      & & {\bullet} \ar@{-}[ruu]_{b} & &     & & {\bullet} \ar@{->}[ruu]_{b} & 
  }}
  \caption{
    The four kinds of clustering coefficients.  a) Regular
    clustering coefficient.  b) Directed clustering coefficient. c) Signed
    clustering coefficient.  d) Signed directed clustering coefficient.
    Edge $c$ is counted when edges $a$ and $b$ are present,
    and for the signed variants, weighted by $\mathrm{sgn}(abc)$.
  }
  \label{fig:clusco}
\end{figure}

The clustering coefficient is defined as the proportion of all incident
edge pairs that are completed by a third edge to form a triangle.
Figure~\ref{fig:clusco} gives an illustration. 
Given an undirected, unsigned graph $G=(V,E)$
its clustering coefficient is given by
\begin{eqnarray}
  c(G) &=& \frac
  {\{ (u, v, w) \in V^3 \mid u \sim v \sim w \sim u \}}
  {\{ (u, v, w) \in V^3 \mid u \sim v \sim w \}} 
  \label{eq:cc}
\end{eqnarray}
To extend the clustering coefficient to negative edges, we assume
structural balance for two incident signed edges.
As shown in Figure~\ref{fig:clusco}, an edge with sign $c$ completing two incident
edges with signs $a$ and $b$ to form a triangle must fulfill the equation $c=ab$. 
\begin{eqnarray}
  c_{\mathrm s}(G) &=& \frac
  { \sum_{u \sim v \sim w \sim u} \sigma(\{u,v\}) \sigma(\{v,w\}) \sigma(\{w,u\}) }
  {\{ u, v, w \in V \mid u \sim v \sim w \}}
  \label{eq:ccs1}
\end{eqnarray}
Therefore, the signed clustering coefficient denotes to what extent
the graph exhibits a balanced structure. 
In actual signed social networks, we expect it to be positive.

Additionally, we define the relative signed clustering coefficient as the quotient
of the signed and unsigned clustering coefficients.  
\begin{eqnarray}
  S(G) = \frac{c_s(G)}{c(G)} = 
  \frac
      { \sum_{u \sim v \sim w \sim u} \sigma(\{u,v\}) \sigma(\{v,w\}) \sigma(\{w,u\}) }
      {\{ u, v, w \in V \mid u \sim v \sim w \sim u \}}
      \label{eq:ccs2} 
\end{eqnarray}
The relative signed clustering coefficient takes on values between $-1$
and $+1$. It is $+1$ when all triangles are balanced.  
In networks with
negative relative signed clustering coefficients, structural balance
does not hold. 
In fact, the relative signed clustering coefficient is closely related
to the number of balanced and unbalanced triangles in a network.  If
$\Delta^+(G)$ is the number of balanced triangles and $\Delta^-(G)$ is the
number of unbalanced triangles in a signed network $G$, then
\begin{eqnarray*}
  S(G) &=& \frac {\Delta^+(G) - \Delta^-(G)} {\Delta^+(G) +
    \Delta^-(G)}. 
\end{eqnarray*}

The directed signed clustering coefficient and directed relative signed clustering
coefficient can be defined analogously using Expressions~(\ref{eq:ccs1})
and~(\ref{eq:ccs2}).  
The signed clustering coefficient and relative signed clustering coefficient are
zero in random networks, when the sign of edges is distributed equally. 
The signed clustering coefficients are by definition smaller than
their unsigned counterparts.

\begin{table}
  \caption{
  \label{tab:clusco}
    The values for all variants of the clustering
    coefficient for the example datasets. 
    The directed variants are not computed for the two undirected
    datasets. 
  }
  \centering
    \begin{tabular}{l|r r r|r r r}
      \toprule
      \textbf{Network} & \multicolumn{3}{c|}{\textbf{Undirected}} &
      \multicolumn{3}{c}{\textbf{Directed}} \\
      & $c(G)$ & $c_{\mathrm s}(G)$ & $S(G)$ &
      $c(G)$ & $c_{\mathrm s}(G)$ & $S(G)$ \\
      \midrule
      Slashdot Zoo & 0.0318 & 0.00607 & 19.1\% & 0.0559 & 0.00918 & 16.4\% \\
      Epinions & 0.1107 & 0.01488 & 13.4\% & 0.1154 & 0.01638 & 14.2\% \\
      %% Libimseti.cz & 0.0075 & -0.00088 & -11.8\% &
      %% 0.0291 & 0.00740 & 25.4\% \\
      Wikipedia elections & 0.1391 & 0.01489 & 10.9\% & 0.1654 & 0.02427 & 14.7\% \\
      Wikipedia conflicts & 0.0580 & 0.03342 & 57.6\% & -- & -- & -- \\
      \midrule
      Highland tribes & 0.5271 & 0.30289 & 57.5\% & -- & -- & -- \\
      \bottomrule
    \end{tabular}
\end{table}

Table~\ref{tab:clusco} gives all four variants of the clustering coefficient
measured in the example datasets, along with the relative signed
clustering coefficients.   
%% We also give the clustering coefficient of a random graph of the same
%% size, as described in~\cite{b228}.  
%% The clustering coefficient of the
%% Slashdot Zoo is significantly larger than that of a random graph of
%% equal size.  Together with the observation in Table~\ref{tab:net}
%% that the average distance between nodes is less than in that of a random
%% graph, we follow Watts and Strogatz~\cite{b228} and conclude that the
%% Slashdot Zoo exhibits the small-world phenomenon. 
The high values for the relative clustering coefficients show that
our multiplication rule is valid in the examined datasets, and justifies
the structural balance approach.

\section{Visualizing Structural Balance:  \\ Signed Graph Drawing}
\label{sec:drawing}
To motivate the use of algebraic graph theory based on structural
balance, we consider the problem of 
drawing signed graphs and show how it naturally leads to our definition
of the Laplacian matrix for signed graphs.  
We begin by showing how the signed Laplacian matrix arises
naturally in the task of drawing graphs with negative edges when one
tries to place each node near to its positive neighbors and opposite to
its negative neighbors, extending a standard method of graph drawing in
the presence of only positive edges.  

The Laplacian matrix turns up in
graph drawing when we try to find an embedding of a graph into a plane in
a way that adjacent nodes are drawn near to each other~\cite{b287}.
In the literature, signed graphs have been drawn using eigenvectors of
the signed adjacency matrix~\cite{b400}.  
Instead, our approach consists of using the Laplacian to draw signed
graphs, in analogy with the unsigned case. 
To do this, we will stipulate that negative edges should be drawn as far
from each other as possible.

\subsection{Unsigned Graphs}
We now describe the general method for generating an embedding of the
nodes of an unsigned graph into the plane using the Laplacian matrix.
Let $G=(V,E)$ be a connected unsigned graph with adjacency matrix
$\mathbf A$. 
We want to find a two-dimensional drawing of $G$
in which each vertex is drawn near to its neighbors.  This requirement
gives rise to the following vertex equation, which states that every
vertex is placed at the mean of its neighbors' coordinates, weighted by
the sign of the connecting edges.  Let $\mathbf X \in
\mathbb{R}^{n\times 2}$ be a matrix whose columns are the coordinates of
all nodes in the drawing, then we have for each node $u$: 
\begin{eqnarray}
  \mathbf X_{u \bullet} = \left(\sum_{u \sim v} \mathbf A_{uv}\right)^{-1}
  \sum_{u \sim v}
  \mathbf A_{uv} \mathbf X_{v \bullet} \label{eq:mean}
\end{eqnarray}
Rearranging and aggregating the equation for all $u$ we arrive at
\begin{eqnarray}
  \mathbf D \mathbf X = \mathbf A \mathbf X \label{eq:dx_ax} 
\end{eqnarray}
or
\begin{eqnarray*}
  \mathbf L \mathbf X &=& \mathbf 0. 
\end{eqnarray*}
In other words, the columns of $\mathbf X$
should belong to the null space
of $\mathbf L$, which leads to the degenerate
solution of $\mathbf X_{u \bullet} = \mathbf 1$ for all $u$, i.e., each
$\mathbf X_{u \bullet}$ having
all components equal, as the all-ones vector $\mathbf 1$ is an
eigenvector of $\mathbf L$ with eigenvalue zero.  To exclude that solution, we
require that the columns $\mathbf X$ be orthogonal to
$\mathbf 1$. 
Additionally, to avoid the degenerate solution $\mathbf X_{
  u\bullet}=\mathbf X_{v \bullet}$
for $u \neq v$, we require that all columns of $\mathbf X$ be
orthogonal. 
This leads to $\mathbf X_{u \bullet}$ being the
eigenvectors associated with the two smallest eigenvalues
of $\mathbf L$ different from
zero.  This solution results in a well-known satisfactory embedding of
unsigned graphs.  Such an embedding is related to the resistance
distance (or commute-time distance) between nodes of the
graph~\cite{b287}.

Note that Equation~(\ref{eq:dx_ax}) can also be transformed to $\mathbf X
= \mathbf D^{-1} \mathbf A\mathbf X$, leading to the eigenvectors of the
asymmetric matrix $\mathbf D^{-1} \mathbf A$.  This alternative
derivation is not investigated here. 

\subsection{Signed Graphs}
\label{subsec:general-weighted-graphs}
We now extend the graph drawing method described in the previous section
to graphs with positive and negative edges.  
To adapt Expression~(\ref{eq:mean}) to negative edges, we interpret a
negative edge as an indication that two vertices should be placed on
opposite sides of the drawing.  
Therefore, we take the opposite coordinates $-\mathbf X_{v \bullet}$ of vertices $v$
adjacent to $u$ through a negative edge,
and then compute the mean, as
pictured in Figure~\ref{fig:mean}.  We may call this construction
\emph{antipodal proximity}. 

\begin{figure}
  \centering
  \subfigure[Unsigned graph]{
    \includegraphics[width=\wThree]{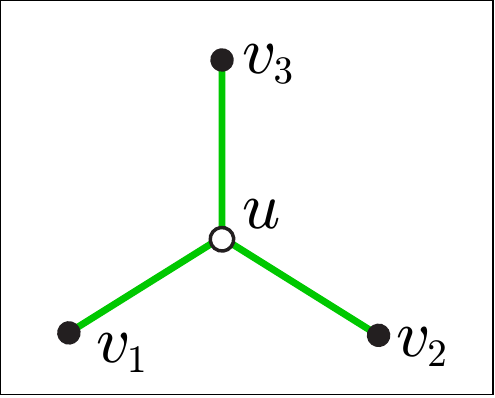}
  }
  \subfigure[Signed graph]{
    \includegraphics[width=\wThree]{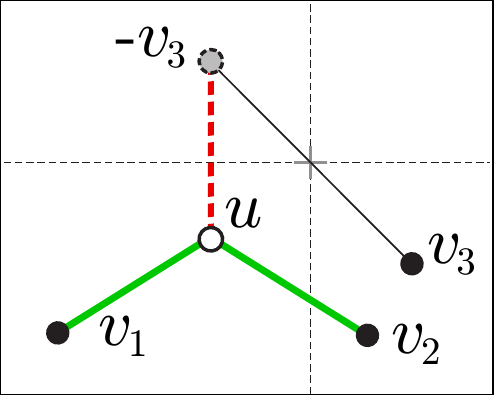}
  }
  \caption{
    Drawing the vertex $u$ at the mean coordinates of its neighbors
    $v_1, v_2, v_3$ by
    proximity and antipodal proximity.
    (a)~In unsigned graphs, a vertex $u$ is placed at the mean of its 
    neighbors $v_1, v_2, v_3$.
    (b)~In signed graphs, a vertex $u$ is
    placed at the mean of 
    its positive neighbors $v_1, v_2$
    and antipodal points $-v_3$ of its negative neighbors. 
  }
  \label{fig:mean}
\end{figure}

This leads to the vertex equation 
\begin{eqnarray}
  \mathbf X_{u \bullet} &= \left( \sum_{\{u,v\}\in E} |\mathbf A_{ij}|\right)^{-1}
  \sum_{\{u,v\}\in E}
  \mathbf A_{uv} \mathbf X_{v \bullet} \label{eq:signed-mean}
\end{eqnarray}
resulting in a signed Laplacian matrix $\mathbf L= \mathbf D-\mathbf A$
in which indeed the definition of the degree matrix $\mathbf D_{uu} =
\sum_v |\mathbf A_{uv}|$ leads to the same equation $\mathbf L \mathbf
X = \mathbf 0$ as in the unsigned case. 

As we will see in the next section, $\mathbf L$ is always
positive-semidefinite, and is positive-definite for graphs
that are unbalanced, i.e., graphs that contain cycles with an odd number of
negative edges.
To obtain a graph drawing from $\mathbf L$, we can thus distinguish three cases,
assuming that $G$ is connected:
\begin{itemize}
\item If all edges are positive, then $\mathbf L$ has one eigenvalue
  zero, and the eigenvectors of the two smallest nonzero eigenvalues can be
  used for graph drawing. 
\item If the graph is unbalanced, $\mathbf L$ is positive-definite and
  we can use the eigenvectors of the two smallest eigenvalues as
  coordinates.  
\item If the graph is balanced, 
  its spectrum is equivalent to that of the
  corresponding unsigned Laplacian matrix, up to signs of the
  eigenvector components.  Using the eigenvectors of the two smallest
  eigenvalues (including zero), we arrive at a graph drawing with all
  points being placed on two parallel lines, reflecting the perfect
  2-clustering present in the graph. 
\end{itemize}

\subsection{Synthetic Examples}
Figure~\ref{fig:toy} shows four small synthetic signed graphs drawn
using the eigenvectors of three characteristic graph matrices.
For each synthetic signed graph, let $\mathbf A$ be its adjacency
matrix, $\mathbf L$ its Laplacian matrix,
and $\mathbf{\bar L}$ the Laplacian matrix of the
corresponding unsigned graph $\bar G=|G|$, i.e., the same graph as $G$, only
that all edges are positive. 
For $\mathbf A$, we use the
eigenvectors corresponding to the two largest absolute eigenvalues. For
$\mathbf L$ and $\mathbf{\bar L}$, we use the eigenvectors of the two smallest
nonzero eigenvalues. 
The small synthetic examples are chosen to display
the basic spectral properties of these three matrices.
All graphs
contain cycles with an odd number of negative edges. 
Column~(a) shows all graphs drawn using the eigenvectors of the
two largest eigenvalues of the adjacency matrix $\mathbf A$.  
%% The
%% eigenvectors of largest absolute eigenvalues are taken of $\mathbf A$ in
%% analogy with graph kernels based on $\mathbf A$, where the largest
%% eigenvalue is mapped to the largest value. 
Column~(b) shows the unsigned Laplacian embedding of the graphs by
setting all edge weights to $+1$.
Column~(c) shows the signed Laplacian embedding.
The embedding given by the eigenvectors of $\mathbf A$ is clearly not
satisfactory for graph drawing.  As expected, the graphs drawn using the
ordinary Laplacian matrix place nodes connected by a negative edge near
to each other.
The signed Laplacian matrix produces a graph embedding where
negative links span large distances across the drawing, as required. 

\begin{figure}
  \centering
  \begin{tabular}{r c c c}

    \raisebox{\raiseToy}{(1)} &
    \includegraphics[width=\wToy]{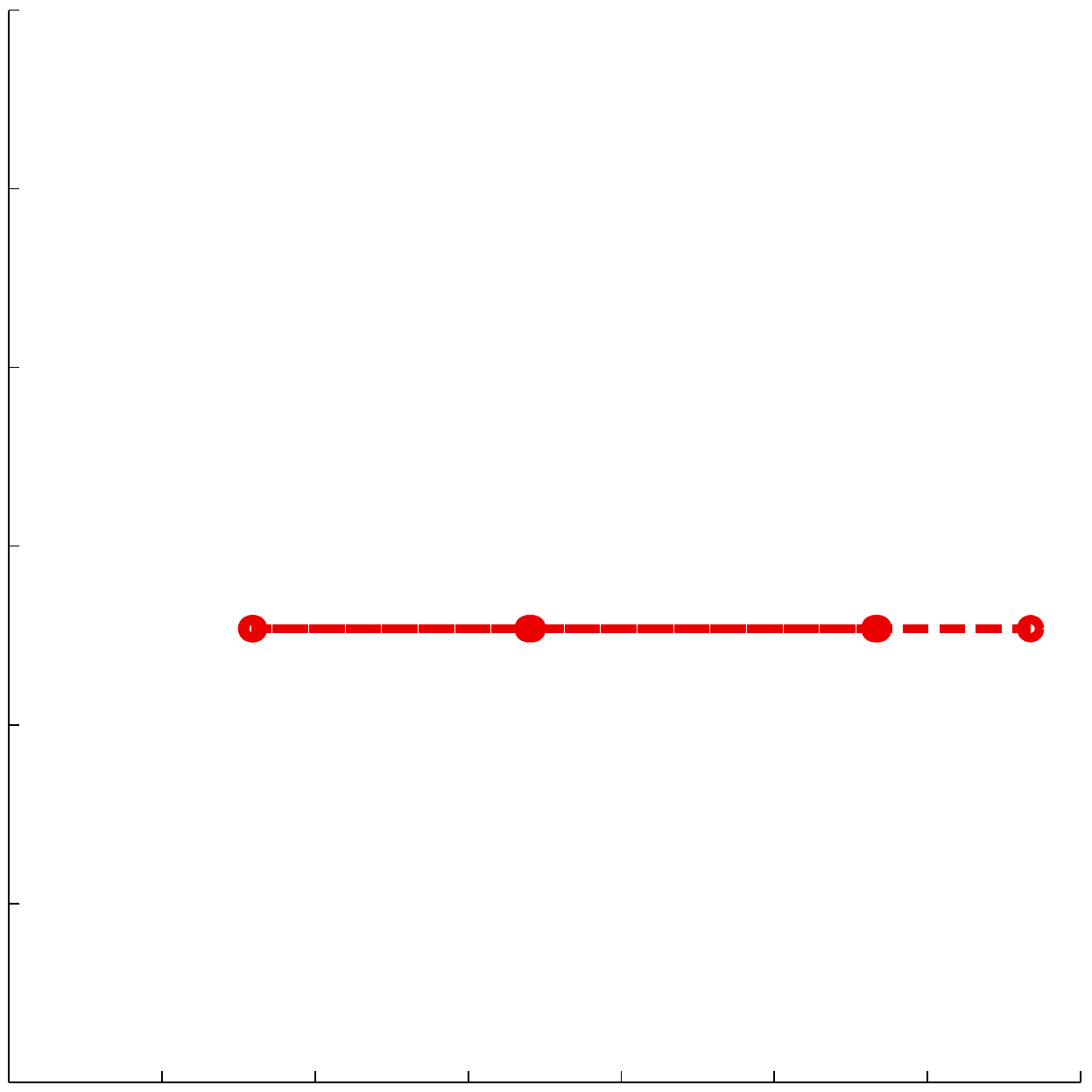} &
    \includegraphics[width=\wToy]{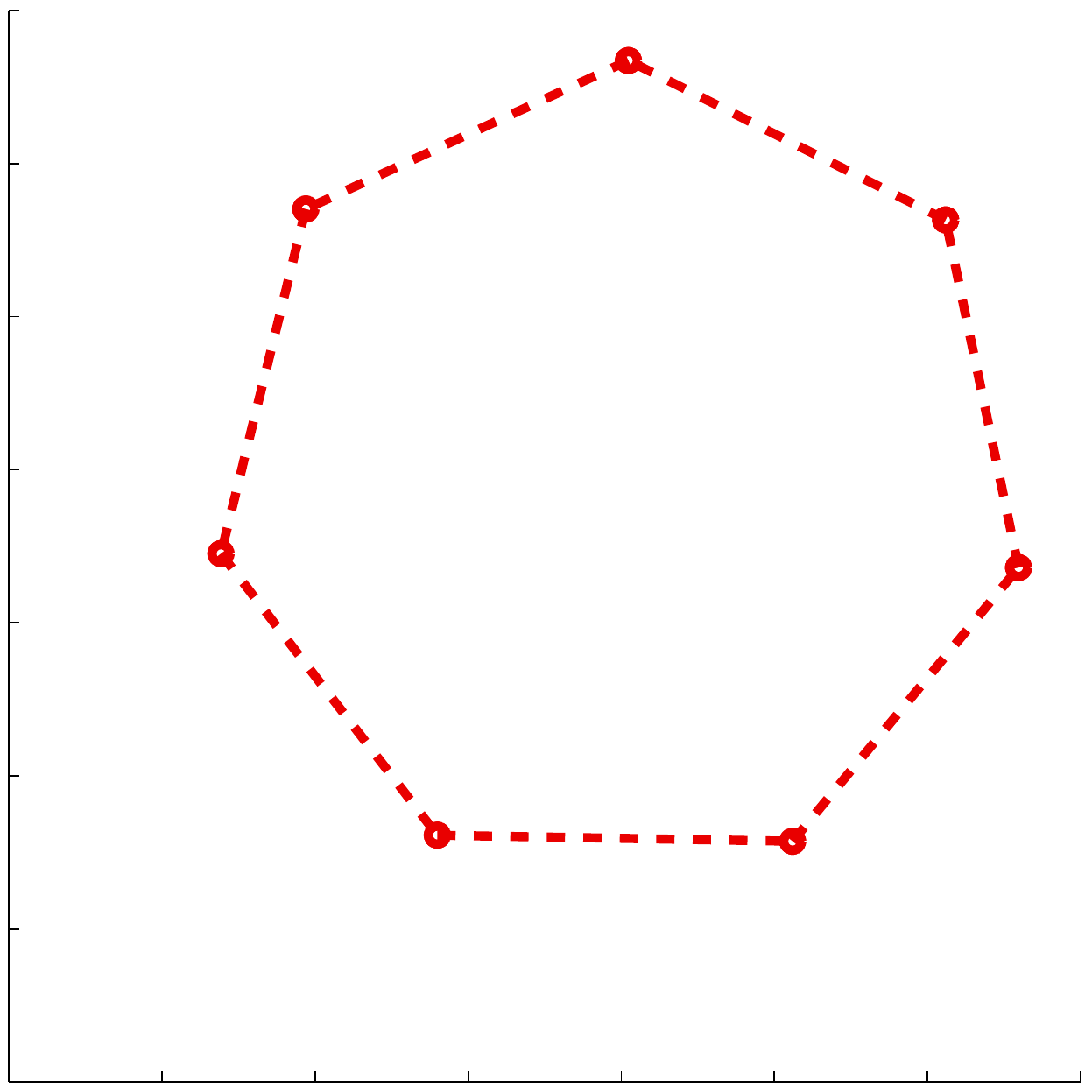} &
    \includegraphics[width=\wToy]{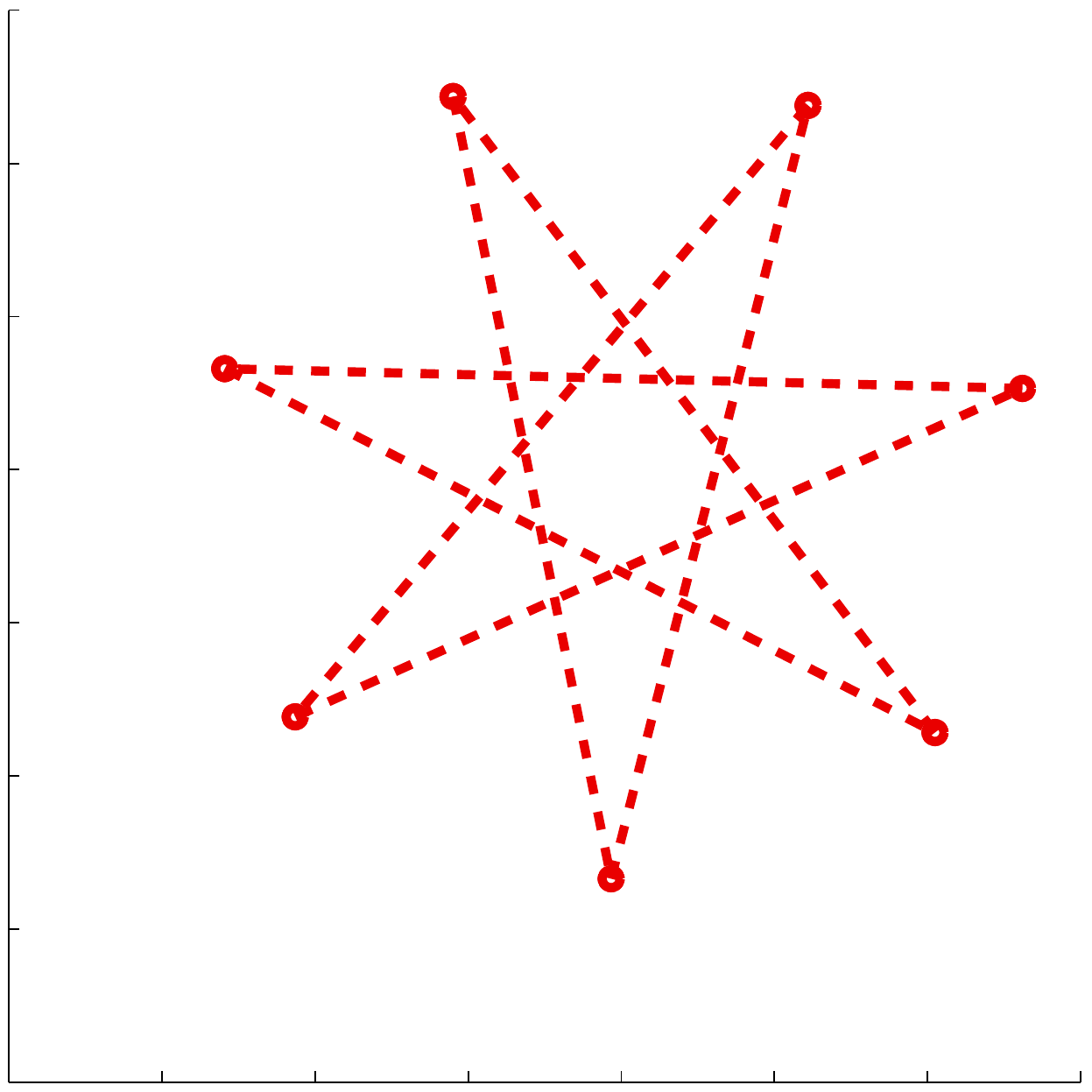} \\

    \raisebox{\raiseToy}{(2)} &
    \includegraphics[width=\wToy]{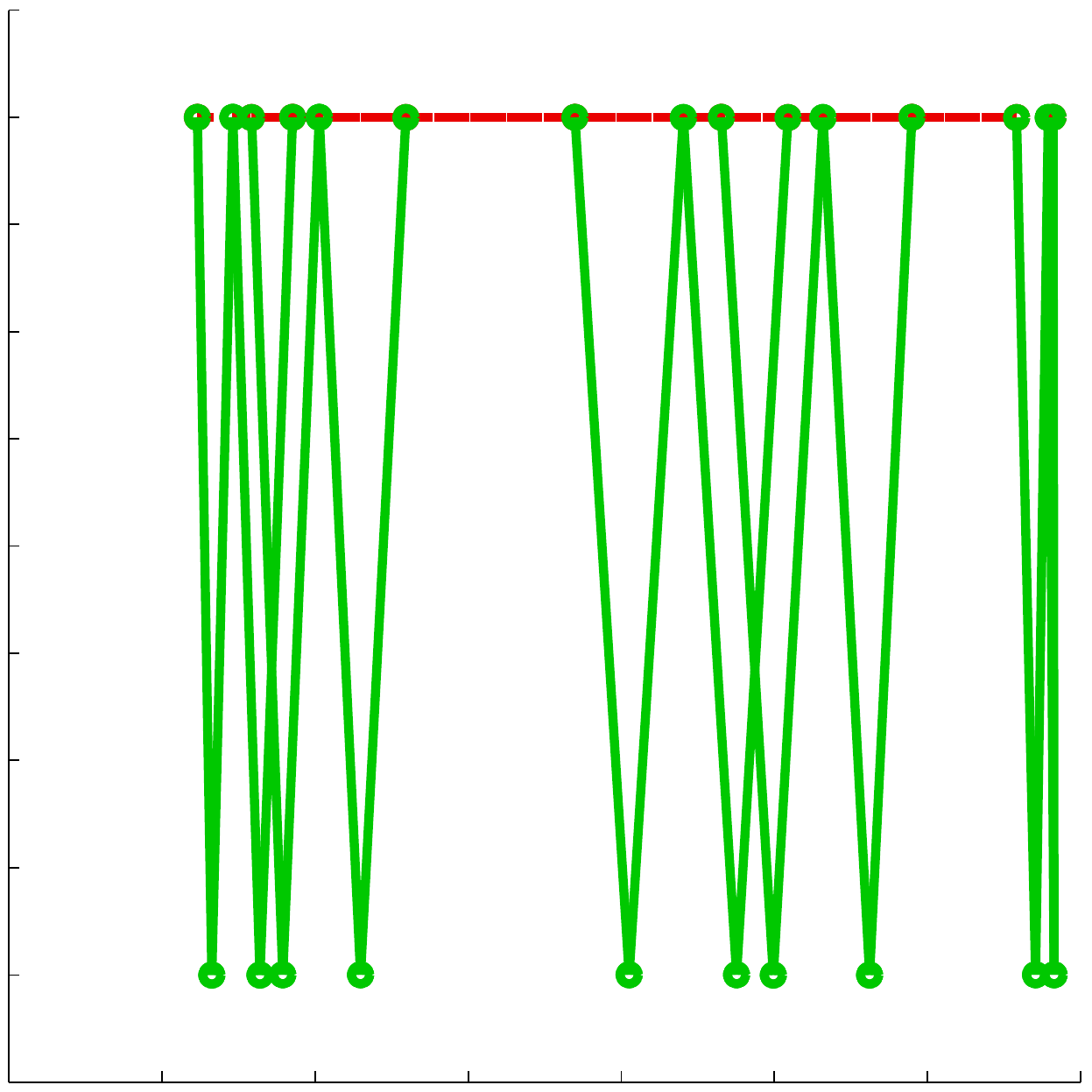} &
    \includegraphics[width=\wToy]{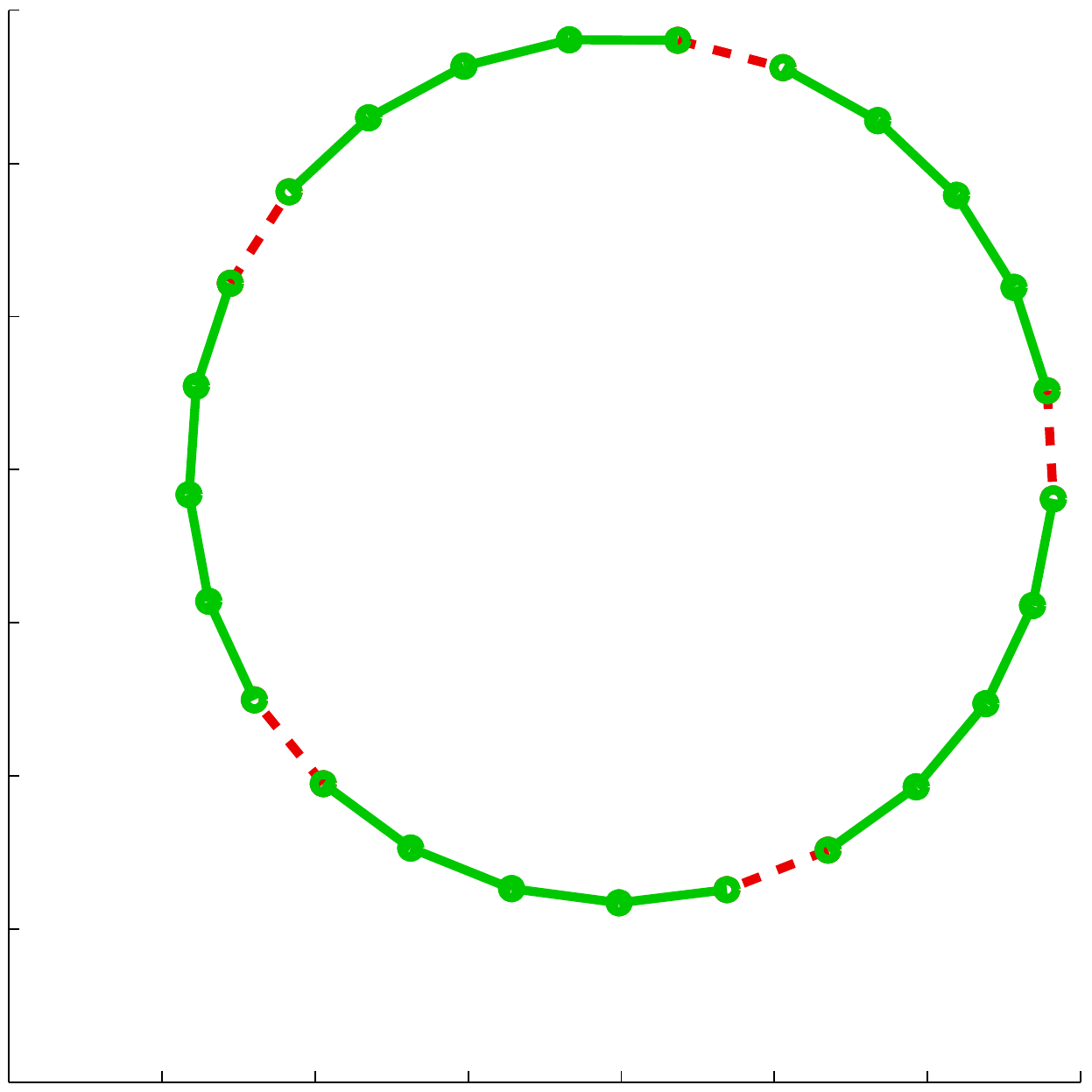} &
    \includegraphics[width=\wToy]{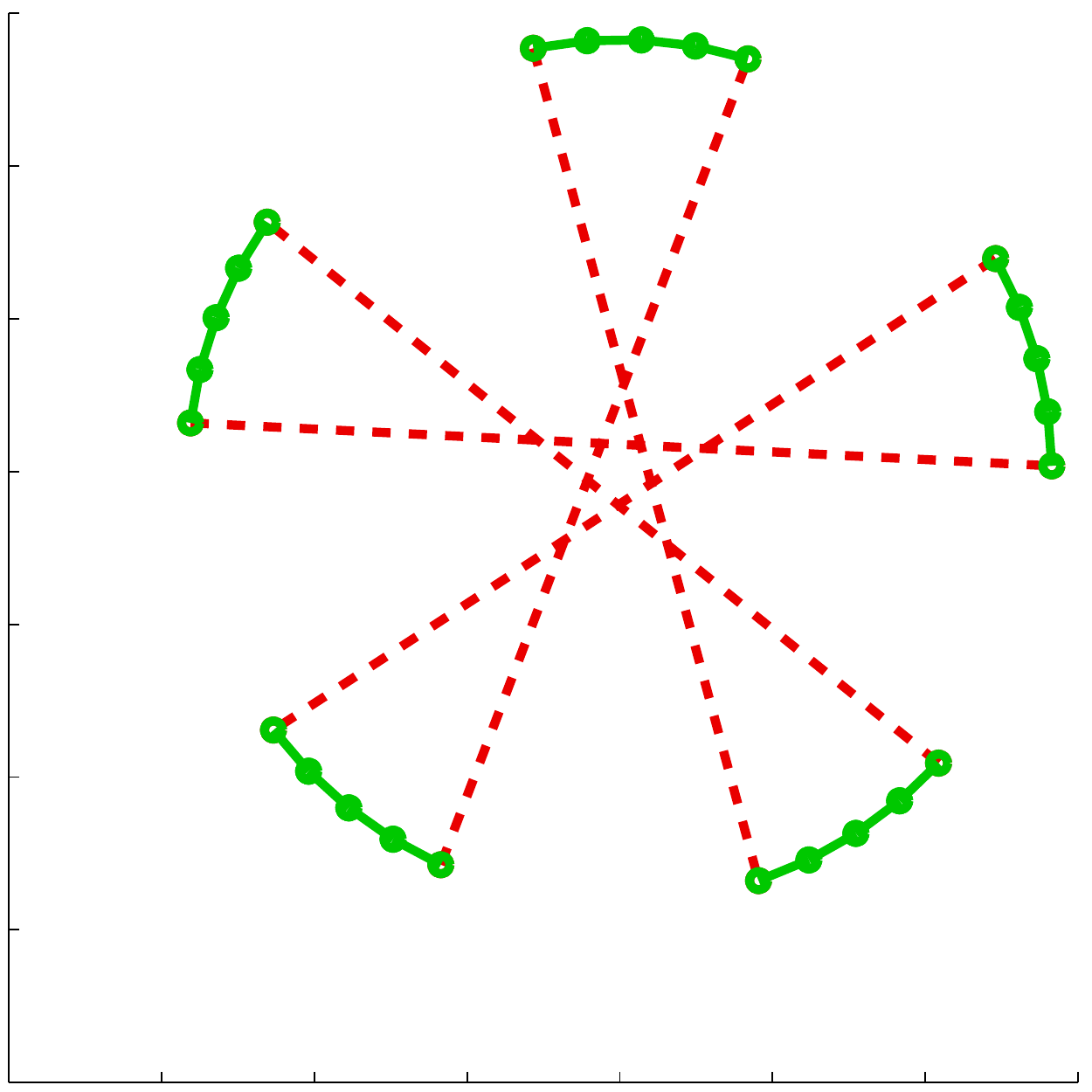} \\

    \raisebox{\raiseToy}{(3)} &
    \includegraphics[width=\wToy]{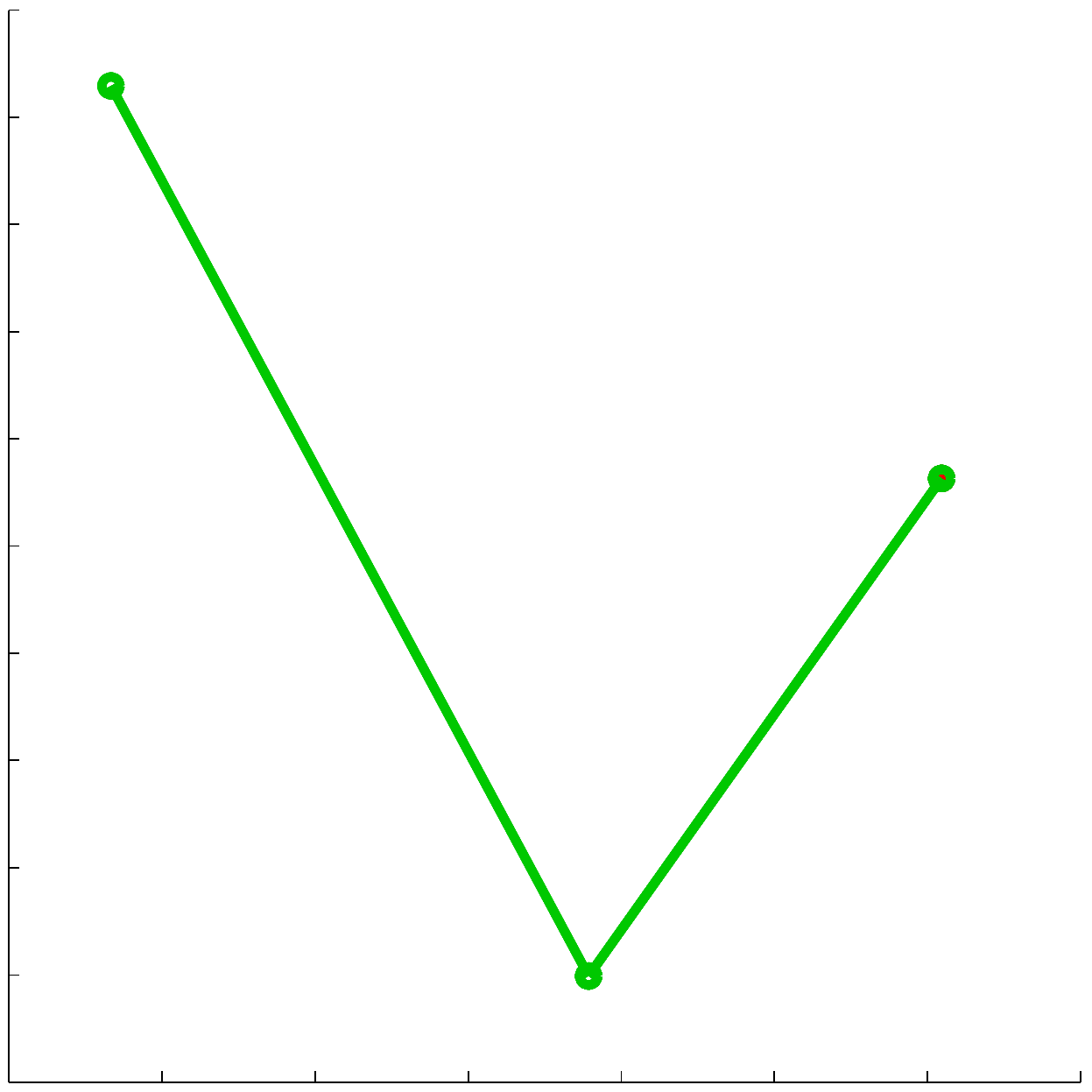} &
    \includegraphics[width=\wToy]{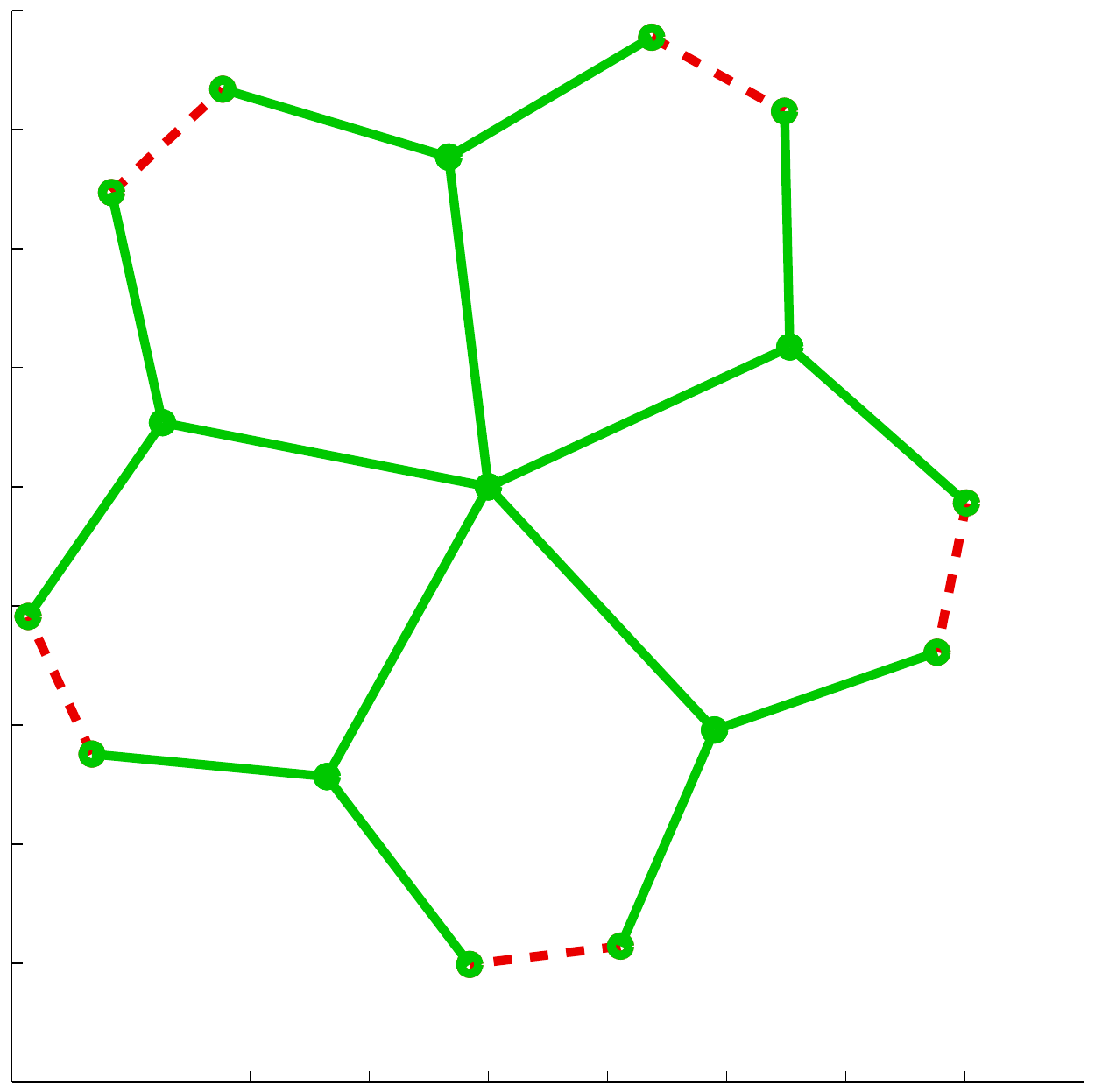} &
    \includegraphics[width=\wToy]{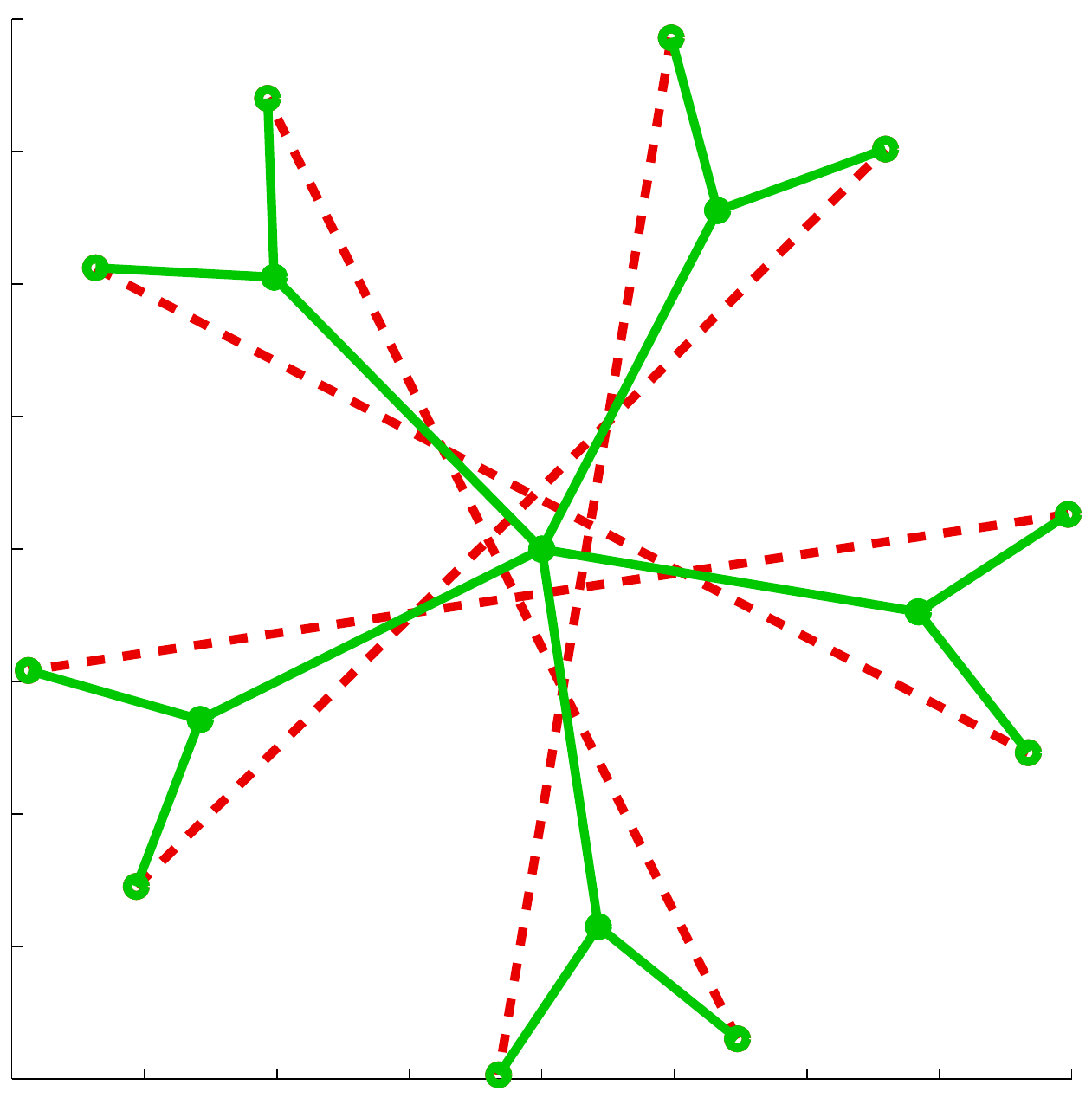} \\

    \raisebox{\raiseToy}{(4)} &
    \includegraphics[width=\wToy]{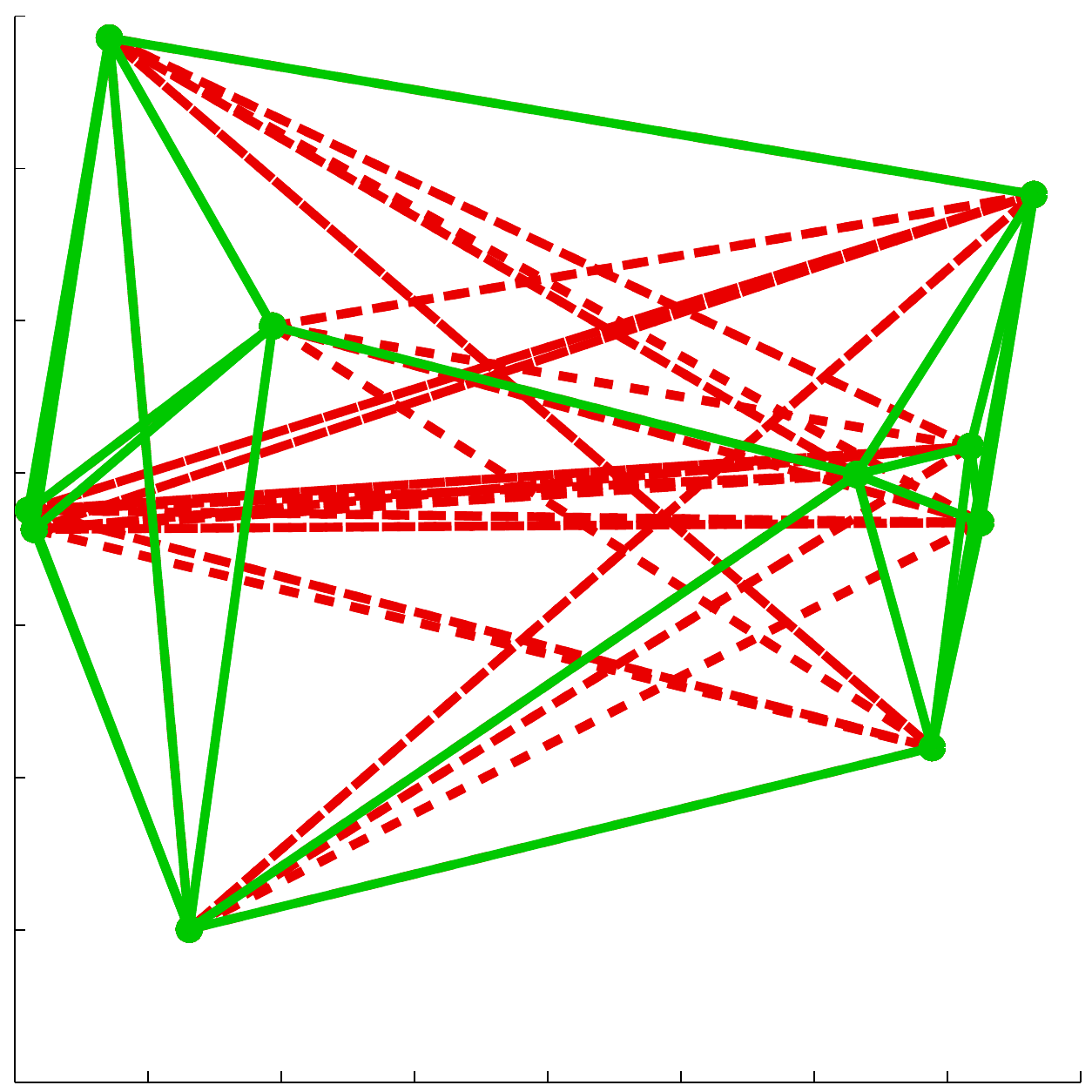} &
    \includegraphics[width=\wToy]{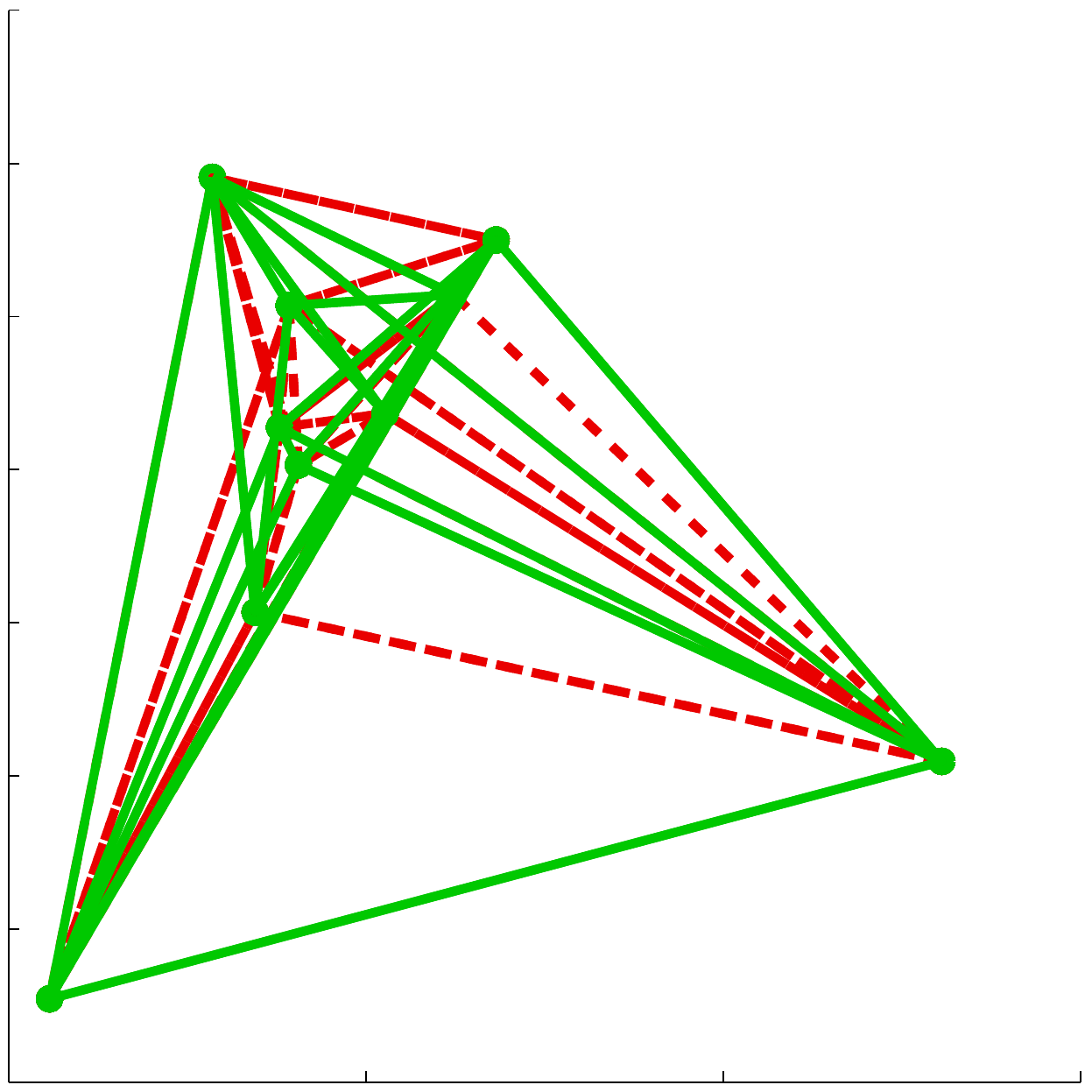} &
    \includegraphics[width=\wToy]{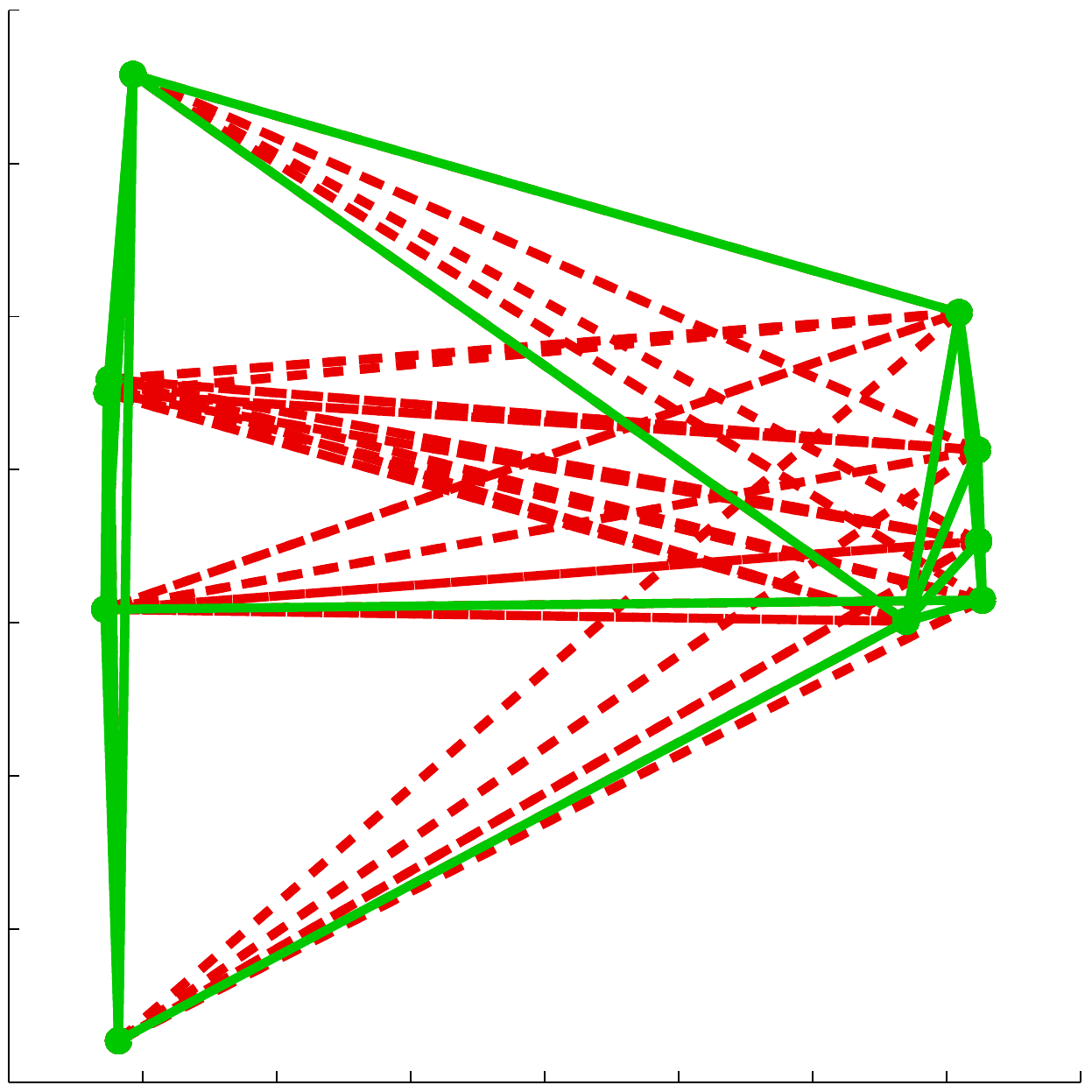} \\
    
    & (a) $\mathbf A$ & (b) $\mathbf{\bar L}$ & (c) $\mathbf L$
  \end{tabular}

  \caption{
    Four small synthetic signed graphs [(1)--(4)] drawn using the
    eigenvectors of three graph matrices. 
    (a) the adjacency matrix~$\mathbf A$,  (b) the Laplacian
    $\mathbf{\bar L}$ of the underlying unsigned
    graph $\bar G$,  (c) the Laplacian $\mathbf L$. 
    All graphs shown contain
    negative cycles, and their signed Laplacian matrices are 
    positive-definite. 
    Positive edges are shown as solid green lines and
    negative edges as red dashed lines. 
  }
  \label{fig:toy}
\end{figure}

\paragraph{Drawing a Balanced Graph}
We assume, in the derivation above, that the eigenvectors corresponding
to the two smallest eigenvalues of the Laplacian $\mathbf L$ should be
used for graph drawing. This is true in that it gives the best possible
drawing according to the proximity and distance criteria of
positive and negative edges. 
If however the graph is balanced, then, as we will see, the smallest eigenvalue of $\mathbf L$ is
zero.  Unlike 
the case in unsigned graphs however, the corresponding eigenvector is not constant
but contains values $\{\pm 1\}$ describing the split into two partitions.
If we use that eigenvector to draw the graph, the resulting drawing
will place all vertices on two lines.  Such an embedding may be
satisfactory in cases where the perfect balance of the graph is to be
visualized.  If however positive edges among each partition's vertices
are also to be visible, the eigenvector corresponding to the third smallest
eigenvalue can be added with a small weight to the first eigenvector,
resulting in a two-dimensional representation of a 3-dimensional
embedding.  
The resulting three methods are illustrated in
Figure~\ref{fig:balanced-solution}.  

\begin{figure}
  \centering
  \subfigure[$(\lambda_1, \lambda_2)$]{
    \includegraphics[width=\wThree]{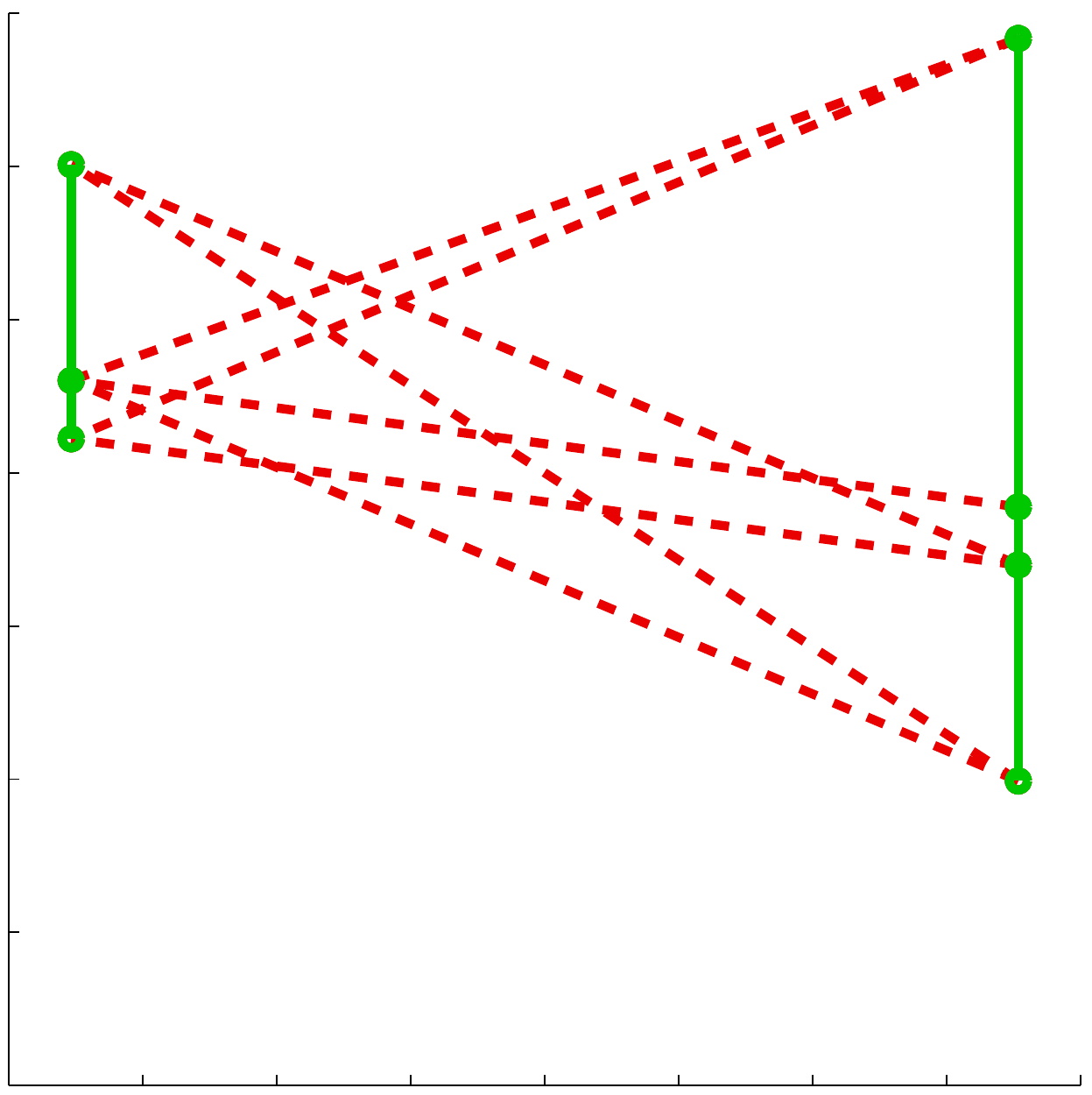}
  }
  \subfigure[$(\lambda_2, \lambda_3)$]{
    \includegraphics[width=\wThree]{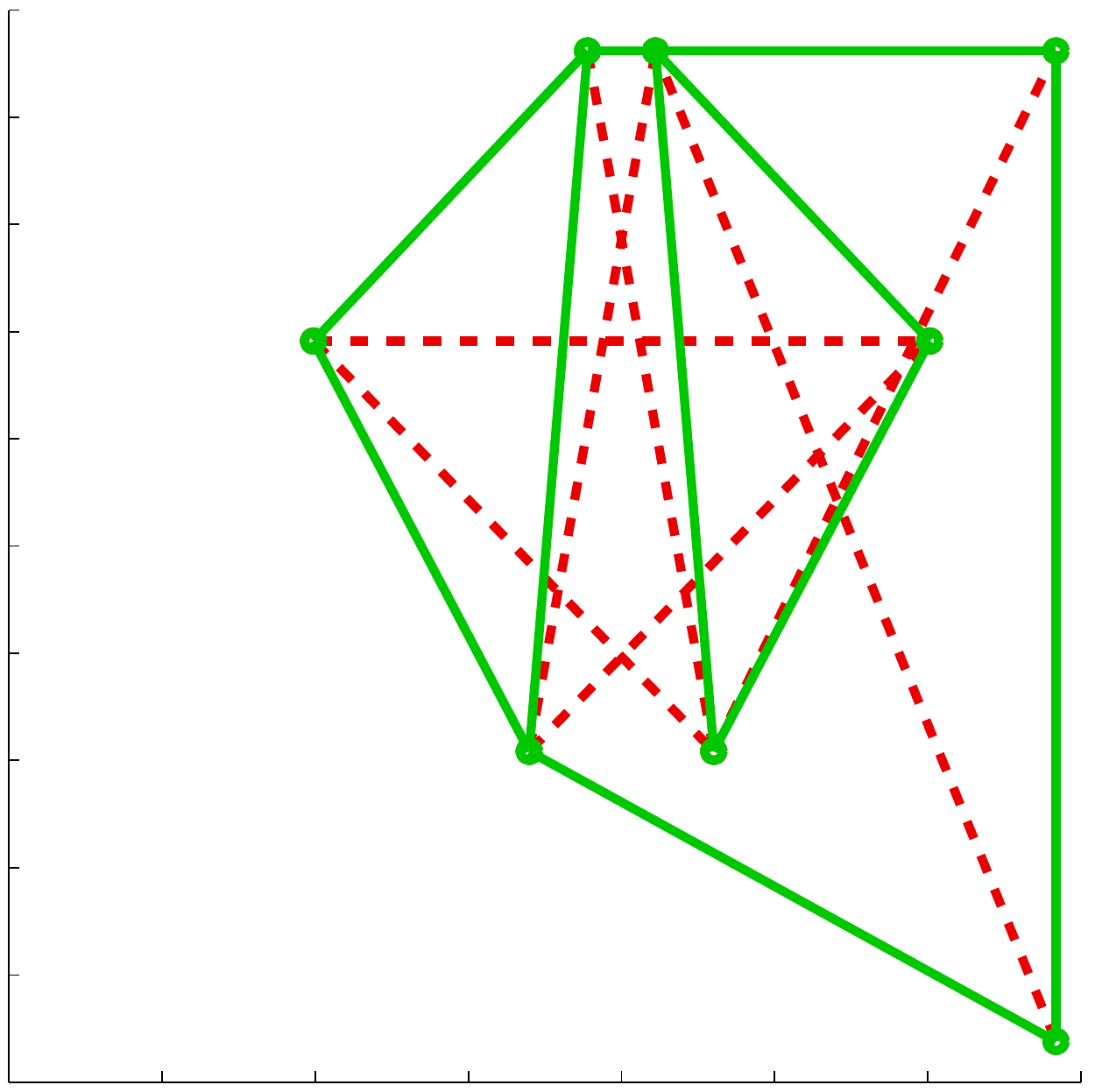}
  }
  \subfigure[$(\lambda_1+0.3\lambda_3, \lambda_2)$]{
    \includegraphics[width=\wThree]{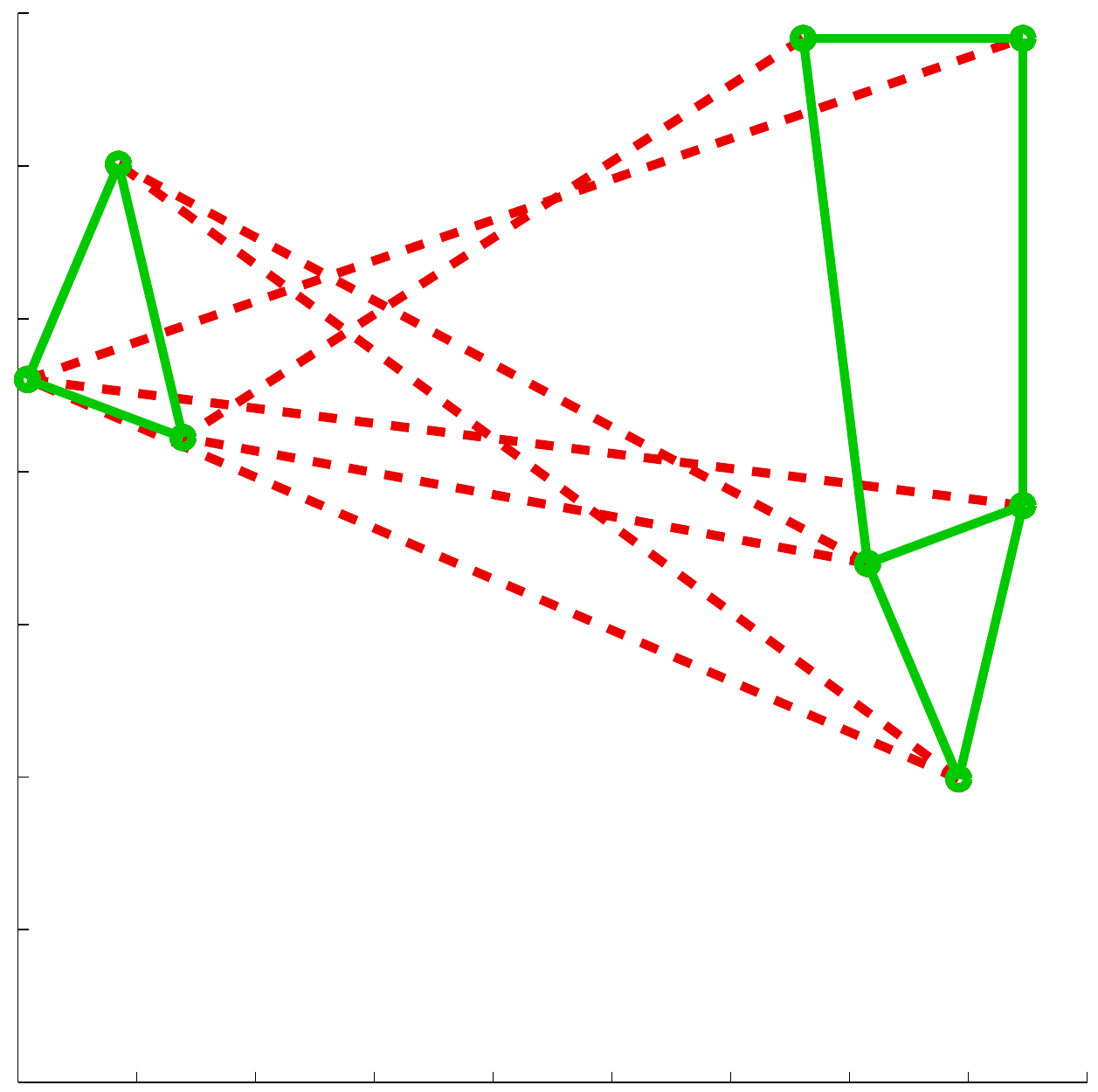}
  }
  \caption{
    Three methods for drawing a balanced signed graph, using a small
    artificial example network.  
    (a) Using the eigenvector corresponding to the smallest eigenvalue
    $\lambda_1 = 0$, intra-cluster structure is lost.  
    (b) Ignoring the first eigenvalue misses important information about
    the clustering.  
    (c) Using a linear combination of both methods gives a good
    compromise.  
  }
  \label{fig:balanced-solution}
\end{figure}

In practice, large graphs are almost always unbalanced as shown in
Figure~\ref{fig:spectra} and Table~\ref{tab:smallest}, and the two
smallest eigenvalues give a satisfactory embedding. 
Figure~\ref{fig:large} shows large signed networks drawn using the two
eigenvectors of the smallest eigenvalues of the Laplacian matrix
$\mathbf L$ for three signed networks. 

\begin{figure}
  \centering
  \subfigure[Slashdot Zoo]{
    \includegraphics[width=\wThree]{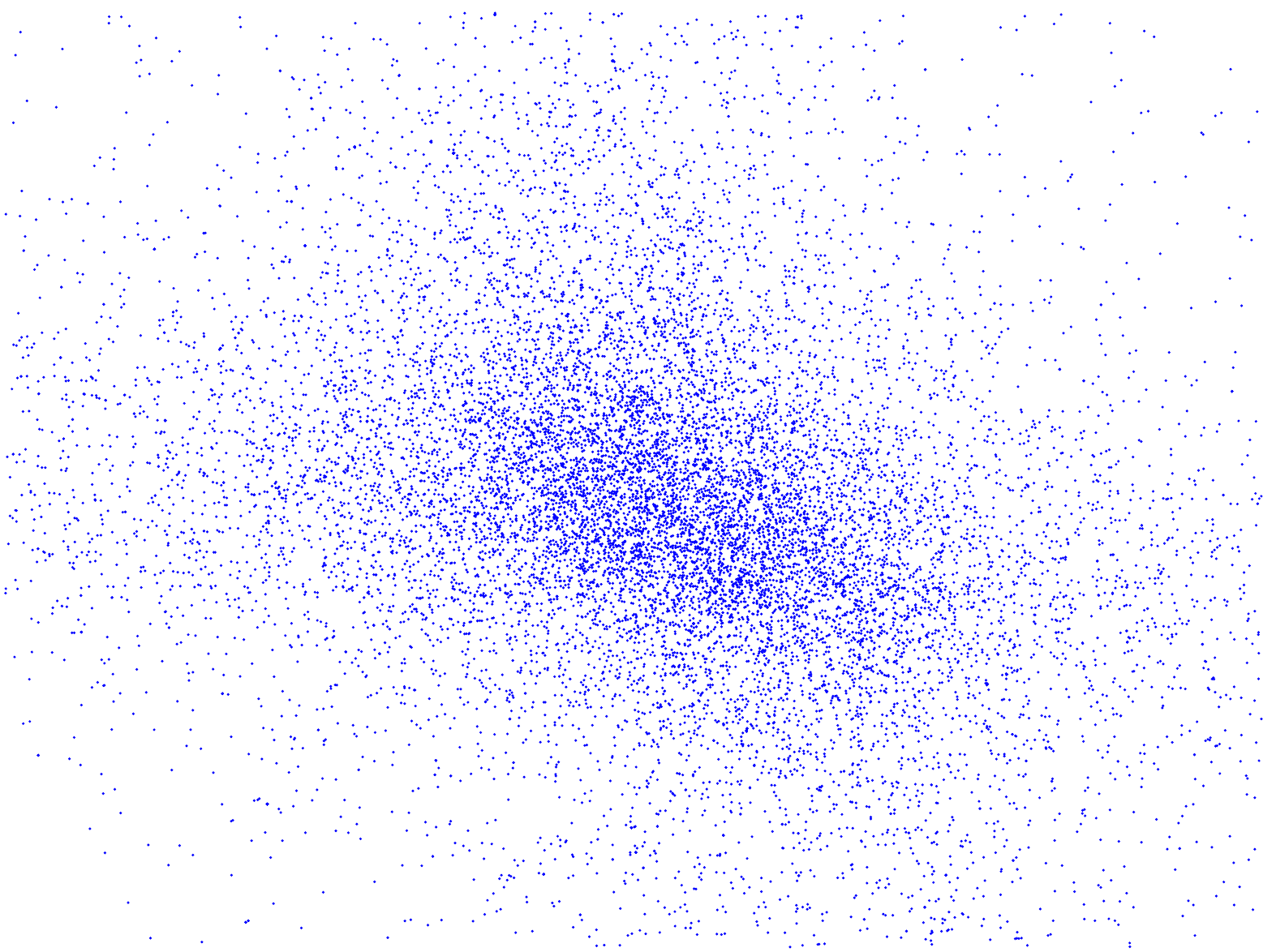}}
  \subfigure[Epinions]{ 
    \includegraphics[width=\wThree]{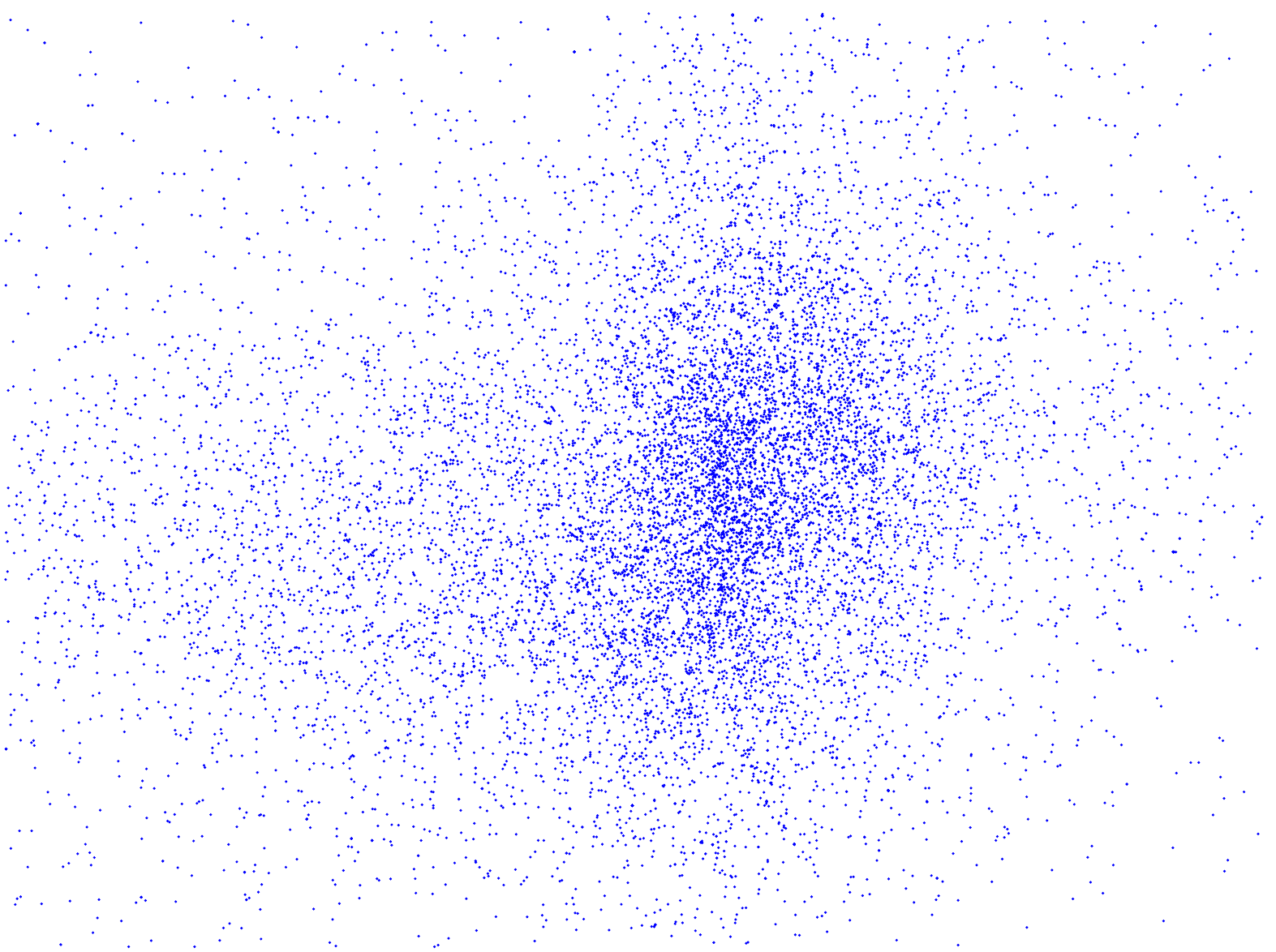}}
  \subfigure[Wikipedia elections]{
    \includegraphics[width=\wThree]{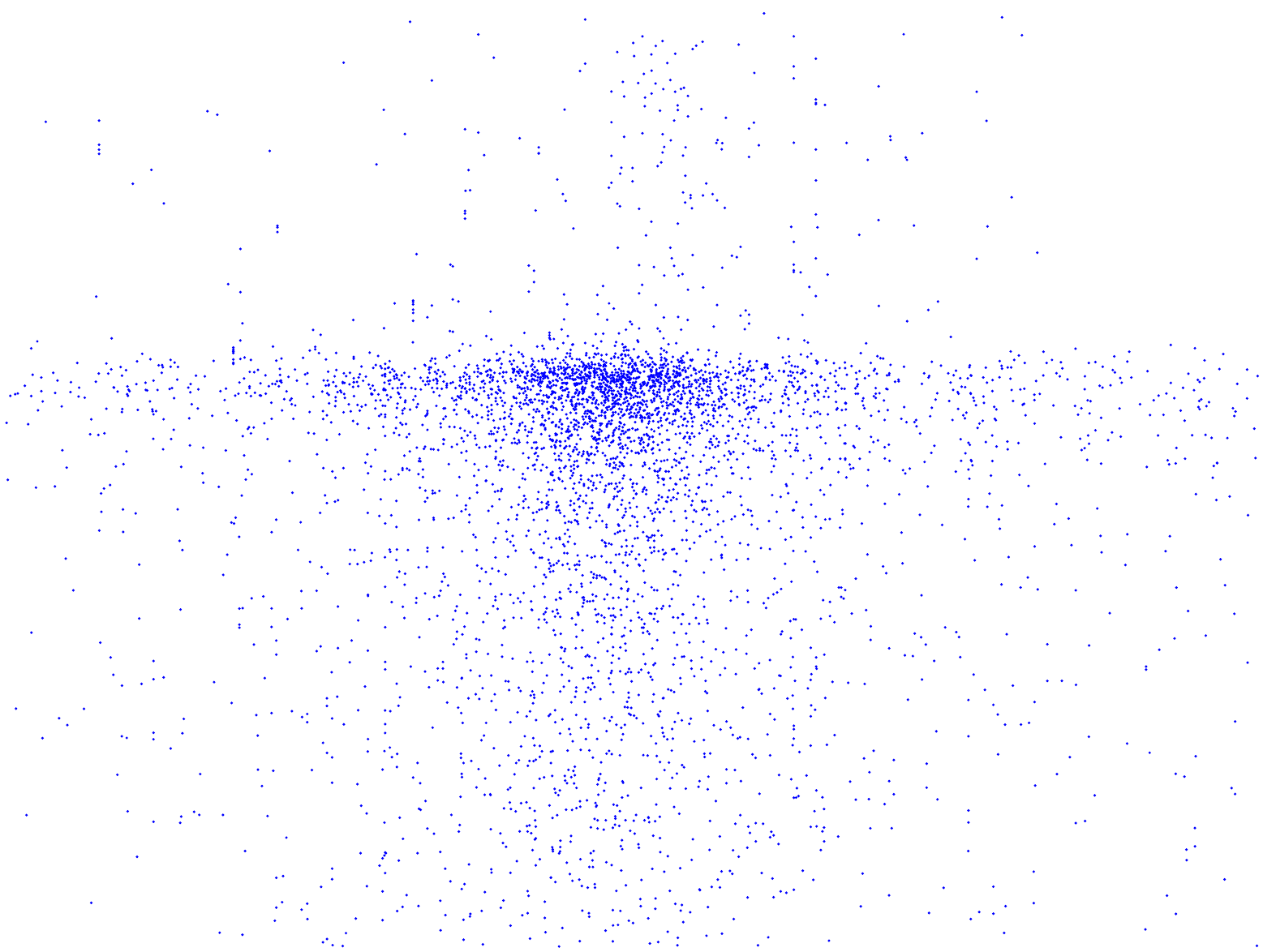}}
  \subfigure[Wikipedia conflicts]{
    \includegraphics[width=\wThree]{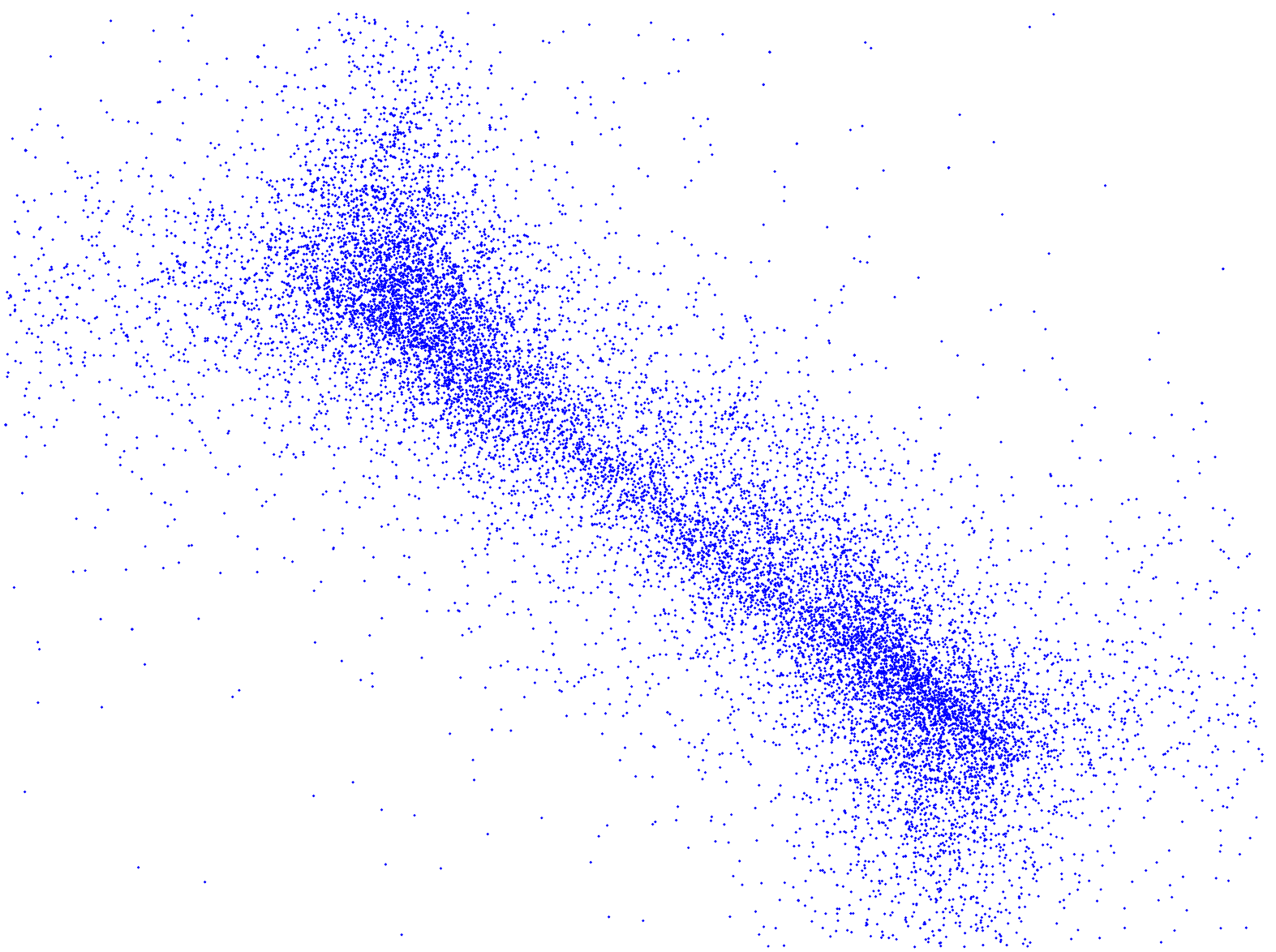}}
  \subfigure[Highland tribes]{
    \includegraphics[width=\wThree]{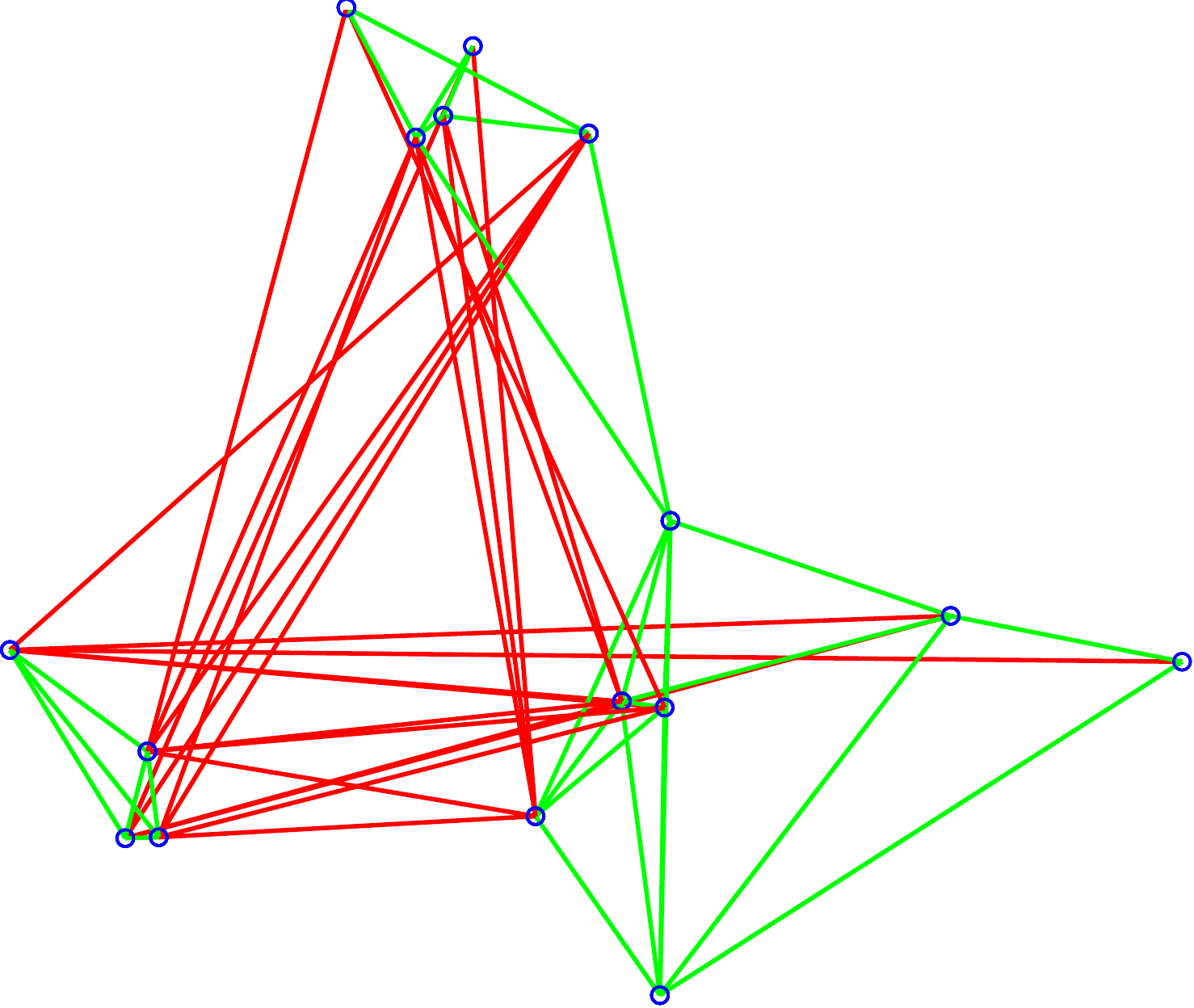}}
  \caption{
    Signed spectral embedding of networks.  For each network, every
    node is represented as a point whose coordinates are the
    corresponding values in the eigenvectors of the signed Laplacian
    $\mathbf L$ corresponding to the two smallest eigenvalues. 
    For the Highland tribes network, positive edges are shown in green
    and negative edges in red. The edges in the other networks are not
    shown. 
  }
  \label{fig:large}
\end{figure}

\section{Capturing Structural Balance:  \\ The Signed Laplacian}
\label{sec:laplacian}
The spectrum Laplacian matrix $\mathbf
L=\mathbf D-\mathbf A$ of signed networks is studied in~\cite{b351}, where it
is established that the signed Laplacian is positive-definite when each
connected component of a graph contains a cycle with an odd number of
negative edges. 
Other basic properties of the Laplacian matrix for signed graphs are given
in~\cite{b356}. 

For an unsigned graph, the Laplacian $\mathbf L$ is positive-semidefinite,
i.e., it has only nonnegative eigenvalues.  
In this section, we prove that the Laplacian matrix $\mathbf L$ of a
signed graph
is positive-semidefinite too, characterize the graphs for which it is
positive-definite, and give the relationship between the eigenvalue
decomposition of the signed Laplacian matrix and the eigenvalue
decomposition of the corresponding unsigned Laplacian matrix.  Our
characterization of the smallest eigenvalue of $\mathbf L$ in
terms of graph balance is based on~\cite{b351}.

\subsection{Positive-semidefiniteness of the Laplacian}
A Hermitian matrix is positive-semidefinite when all its eigenvalues are
nonnegative, and positive-definite when all its eigenvalues are
positive. 
The Laplacian matrix of an unsigned graph is symmetric and thus
Hermitian. Its smallest eigenvalue is zero, and thus the Laplacian of an
unsigned graph is always positive-semidefinite but never
positive-definite.
In the following, we prove that that Laplacian of a signed graph is
always positive-semidefinite, and positive-definite when the graph is
unbalanced. 

\begin{theorem}
  \label{theo:l-positive-semidefinite}
  The Laplacian matrix $\mathbf L$ of a signed graph
  $G=(V,E,\sigma)$ is positive-semidefinite. 
\end{theorem}
\begin{proof}
  We write the Laplacian matrix as a sum over the edges of $G$:
  \begin{eqnarray*}
    \mathbf L = \sum_{\{u,v\} \in E} \mathbf L^{\{u,v\}}
  \end{eqnarray*}
  where $\mathbf L^{\{u,v\}} \in \mathbb{R}^{|V| \times |V|}$ contains the four
  following nonzero entries: 
  \begin{eqnarray}
    \mathbf L^{\{u,v\}}_{uu} = \mathbf L^{\{u,v\}}_{vv} &=& 1
    \label{eq:single-edge-laplacian} \\
    \mathbf L^{\{u,v\}}_{uv} = \mathbf L^{\{u,v\}}_{vu} &=& -\sigma(\{u,v\}) \nonumber
  \end{eqnarray}
  Let $\mathbf x \in \mathbb{R}^{|V|}$ be a vertex-vector. 
  By considering the bilinear form $\mathbf x^{\mathrm T} \mathbf
  L^{\{u,v\}} \mathbf x$,  
  we see that $\mathbf L^{\{u,v\}}$ is positive-semidefinite:
  \begin{eqnarray}
    \mathbf x^{\mathrm T} \mathbf L^{\{u,v\}} \mathbf x &=& \mathbf x_u^2 + \mathbf x_v^2 -2 \sigma(\{u,v\}) \mathbf x_u
    \mathbf x_v  \label{eq:bilinear} \\ 
    &=& (\mathbf x_u - \sigma(\{u,v\}) \mathbf
    x_v)^2 \nonumber   \\
    &\geq & 0  \nonumber
  \end{eqnarray}
  We now consider the bilinear form $\mathbf x^{\mathrm T} \mathbf L
  \mathbf x$: 
  \begin{eqnarray*}
    \mathbf x^{\mathrm T}\mathbf L \mathbf x = \sum_{\{u,v\}\in E}
    \mathbf x^{\mathrm T} \mathbf L^{\{u,v\}} \mathbf x \ge 0
  \end{eqnarray*}
  It follows that $\mathbf L$ is positive-semidefinite. 
\end{proof}

\newcommand{\Eta}{H}

Another way to prove that $\mathbf L$ is positive-semidefinite consists of
expressing it using the incidence matrix of $G$.  Assume that for each
edge $\{u,v\}$ an arbitrary orientation is chosen.  Then we define the
incidence matrix $\mathbf \Eta \in \mathbb{R}^{|V| \times |E|}$ of $G$ as
\begin{eqnarray*}
  \mathbf \Eta_{u \{u,v\}} &=& 1 \\
  \mathbf \Eta_{v \{u,v\}} &=& -\sigma(\{u,v\}).
\end{eqnarray*}
Here, the letter $\mathbf \Eta$ is the uppercase greek letter
Eta, as used for instance in~\cite{b651}.  
We now consider the product $\mathbf \Eta \mathbf \Eta^{\mathrm T} \in
\mathbb{R}^{|V| \times |V|}$: 
\begin{eqnarray*}
  (\mathbf \Eta \mathbf \Eta^{\mathrm T})_{uu} &=& d(u) \\
  (\mathbf \Eta \mathbf \Eta^{\mathrm T})_{uv} &=& - \sigma(\{u,v\})
\end{eqnarray*}
for diagonal and off-diagonal entries, respectively. 
Therefore $\mathbf \Eta \mathbf \Eta^{\mathrm T} = \mathbf L$, and it follows that
$\mathbf L$ is
positive-semidefinite.  This result is independent of the orientation
chosen for~$\mathbf \Eta$.   

\subsection{Positive-definiteness of $\mathbf L$}
\label{subsec:positive-definiteness}
We now show that, unlike the ordinary Laplacian matrix, the signed
Laplacian matrix is strictly positive-definite for some graphs, including most
real-world networks.
The theorem presented here can be found in~\cite{b356}, and also follows
directly from an earlier result in~\cite{b647}.  

As with the ordinary Laplacian matrix, the spectrum of the signed
Laplacian matrix of a disconnected graph is the union of the spectra of
its connected components. 
This can be seen by noting that the Laplacian matrix of an unconnected
graph has block-diagonal form, with each diagonal entry being the
Laplacian matrix of a single component.  
Therefore, we will restrict the exposition to connected graphs.

Using Definition~\ref{def:balance} of structural balance, we can characterize the graphs for which the
signed Laplacian matrix is positive-definite.
\begin{theorem}
  \label{theo:one}
  The signed Laplacian matrix of an unbalanced graph is positive-definite.   
\end{theorem}

\begin{proof}
  We show that if the bilinear form $\mathbf x^{\mathrm T} \mathbf L
  \mathbf x$ is zero for some vector $\mathbf x \neq \mathbf 0$, then a 
  bipartition of the vertices as described above exists.

  Let $\mathbf x^{\mathrm T} \mathbf L \mathbf x = 0$.  
  We have seen that for every $\mathbf L^{\{u,v\}}$ as defined in
  Equation~(\ref{eq:single-edge-laplacian}) and any $\mathbf   
  x$, $\mathbf x^{\mathrm T} \mathbf L^{\{u,v\}} \mathbf x \geq 0$.
  Therefore, we have for every edge $\{u,v\}$:
  \begin{eqnarray*}
    \mathbf x^{\mathrm T} \mathbf L^{\{u,v\}} \mathbf x &=& 0 \\
    \Leftrightarrow (\mathbf x_u -\sigma(\{u,v\}) \mathbf x_v)^2 &=& 0  \\
    \Leftrightarrow \mathbf x_u &=& \sigma(\{u,v\}) \mathbf x_v 
  \end{eqnarray*}
  In other words, two components of $\mathbf x$ are equal if the corresponding
  vertices are connected by a positive edge, and opposite to each other if
  the corresponding vertices are connected by a negative edge.  
  Because the graph is connected, it follows that all $|\mathbf x_u|$ must be
  equal.  
  We can exclude the solution $\mathbf x_u = 0$ for all $u$ because
  $\mathbf x$ is not
  the zero vector.
  Without loss of generality, we assume that $|\mathbf x_u|=1$ for all $u$. 

  Therefore, $\mathbf x$ gives a bipartition into vertices with $\mathbf x_u = +1$ and
  vertices with $\mathbf x_u = -1$, with the property that two vertices with the
  same value of $\mathbf x_u$ are in the same partition and two vertices with
  opposite sign of $\mathbf x_u$ are in different partitions, and therefore $G$ is
  balanced. 
  Equivalently, the signed Laplacian matrix $\mathbf L$ of a 
  connected unbalanced signed graph is positive-definite. 
\end{proof}

\subsection{Balanced Graphs}
\index{balance}
\index{conflict}
We now show how the spectrum and eigenvectors of the signed Laplacian of
a balanced graph arise from the spectrum and the eigenvalues of the
corresponding unsigned graph by multiplication of eigenvector
components with $\pm 1$. 

Let $G = (V, E, \sigma)$ be a balanced signed graph 
and $\bar G = (V, E)$ the corresponding unsigned graph.
%% Let $\mathbf A$ and $\mathbf{\bar A}$ be the adjacency matrices of $G$
%% and $\bar G$.  
Since $G$ is balanced, there is a vector $\mathbf x \in \{-1,+1\}^{|V|}$ such
that the sign of each edge $\{u,v\}$ is $\sigma(\{u,v\}) = \mathbf x_u \mathbf x_v$. 
\begin{theorem}
  \label{theo:two}
  If $\mathbf L$ is the signed Laplacian matrix of the balanced graph
  $G$ with bipartition $\mathbf x$ and eigenvalue
  decomposition $\mathbf L = \mathbf U \Lambda \mathbf U^{\mathrm T}$, 
  then the eigenvalue decomposition of the Laplacian matrix $\mathbf{\bar L}$
  of $\bar G$
  of the corresponding unsigned graph $\bar G$ of $G$ is given by 
  $\mathbf{\bar L} = \mathbf{\bar U} \Lambda 
  \mathbf{\bar U}^{\mathrm T}$  where
  \begin{eqnarray*}
    \mathbf{\bar U}_{uk} &=& \mathbf x_u \mathbf U_{uk}.
  \end{eqnarray*}
\end{theorem}
\begin{proof}
  To see that $\mathbf{\bar L} = \mathbf{\bar U} \Lambda 
  \mathbf{\bar U}^{\mathrm T}$, note that for
  diagonal elements, we have $\mathbf{\bar U}_{u\bullet}^{\phantom{\mathrm
      I}} \Lambda \mathbf{\bar U}_{u\bullet}^{\mathrm T} = 
  \mathbf x_u^2 \mathbf U_{u\bullet}^{\phantom{\mathrm I}}
  \Lambda \mathbf U_{u\bullet}^{\mathrm T} = \mathbf
  U_{u\bullet}^{\phantom{\mathrm I}} \Lambda \mathbf
  U_{u\bullet}^{\mathrm T} = \mathbf L_{uu} = \mathbf{\bar L}_{uu}$.  
  For off-diagonal elements, we have $\mathbf{\bar
    U}_{u\bullet}^{\phantom{\mathrm I}} \Lambda 
  \mathbf{\bar U}_{v\bullet}^{\mathrm T} = \mathbf x_u \mathbf x_v \mathbf
  U_{u\bullet}^{\phantom{\mathrm I}} \Lambda \mathbf
  U_{v\bullet}^{\mathrm T} = \sigma(\{u,v\}) 
  \mathbf L_{uv} =
  - \sigma(\{u,v\}) \sigma(\{u,v\}) = -1 = \mathbf{\bar L}_{uv}$.  

  We now show that $\mathbf{\bar U} \Lambda 
  \mathbf{\bar U}^{\mathrm T}$ is an eigenvalue
  decomposition of $\mathbf{\bar L}$ by showing that $\mathbf{\bar U}$
  is orthogonal. 
  To see that the columns of $\mathbf{\bar U}$ are indeed orthogonal, note that for any
  two column indexes $k \neq l$, we have $\mathbf{\bar U}_{\bullet k}^{\mathrm T}
  \mathbf{\bar U}_{\bullet l}^{\phantom{\mathrm I}} =
  \sum_{u \in V} \mathbf{\bar U}_{uk} \mathbf{\bar U}_{ul} = \sum_{u \in V} \mathbf x_u^2
  \mathbf U_{uk} \mathbf U_{ul} =
  \mathbf U_{\bullet k}^{\mathrm T} \mathbf U_{\bullet
    l}^{\phantom{\mathrm I}} = 0$ because $\mathbf U$ is orthogonal. 
  Changing signs in $\mathbf U$ does not change the norm of each column vector,
  and thus $\mathbf{\bar L} = \mathbf{\bar U} \Lambda 
  \mathbf{\bar U}^{\mathrm T}$ is the
  eigenvalue decomposition of $\mathbf{\bar L}$.
\end{proof}

As shown in Section \ref{subsec:positive-definiteness}, the Laplacian
matrix of an unbalanced
graph is positive-definite and therefore its spectrum is different
from that of the corresponding unsigned graph. 
Aggregating Theorems~\ref{theo:one} and~\ref{theo:two}, we arrive at the
main result of this section.

\begin{theorem}
  The Laplacian matrix of a connected signed graph is positive-definite if and only
  if the graph is unbalanced.
\end{theorem}
\begin{proof}
  From Theorem~\ref{theo:one} we know that every unbalanced connected
  graph has a positive-definite Laplacian matrix.
  Theorem~\ref{theo:two} implies that every balanced graph has the same
  Laplacian spectrum as its corresponding unsigned graph.  Since the
  unsigned Laplacian is always singular, the signed Laplacian of a
  balanced graph is also singular.  
  Together, these imply that the Laplacian matrix of a connected
  signed graph is positive-definite if and only if the graph is unbalanced. 
\end{proof}

For a general signed graph that need not be connected, we can therefore make
the following statement:  The multiplicity of the eigenvalue zero equals
the number of balanced connected components in $G$~\cite{b651}.  

\begin{figure}
  \subfigure[Slashdot Zoo]{
    \includegraphics[width=\wTwo]{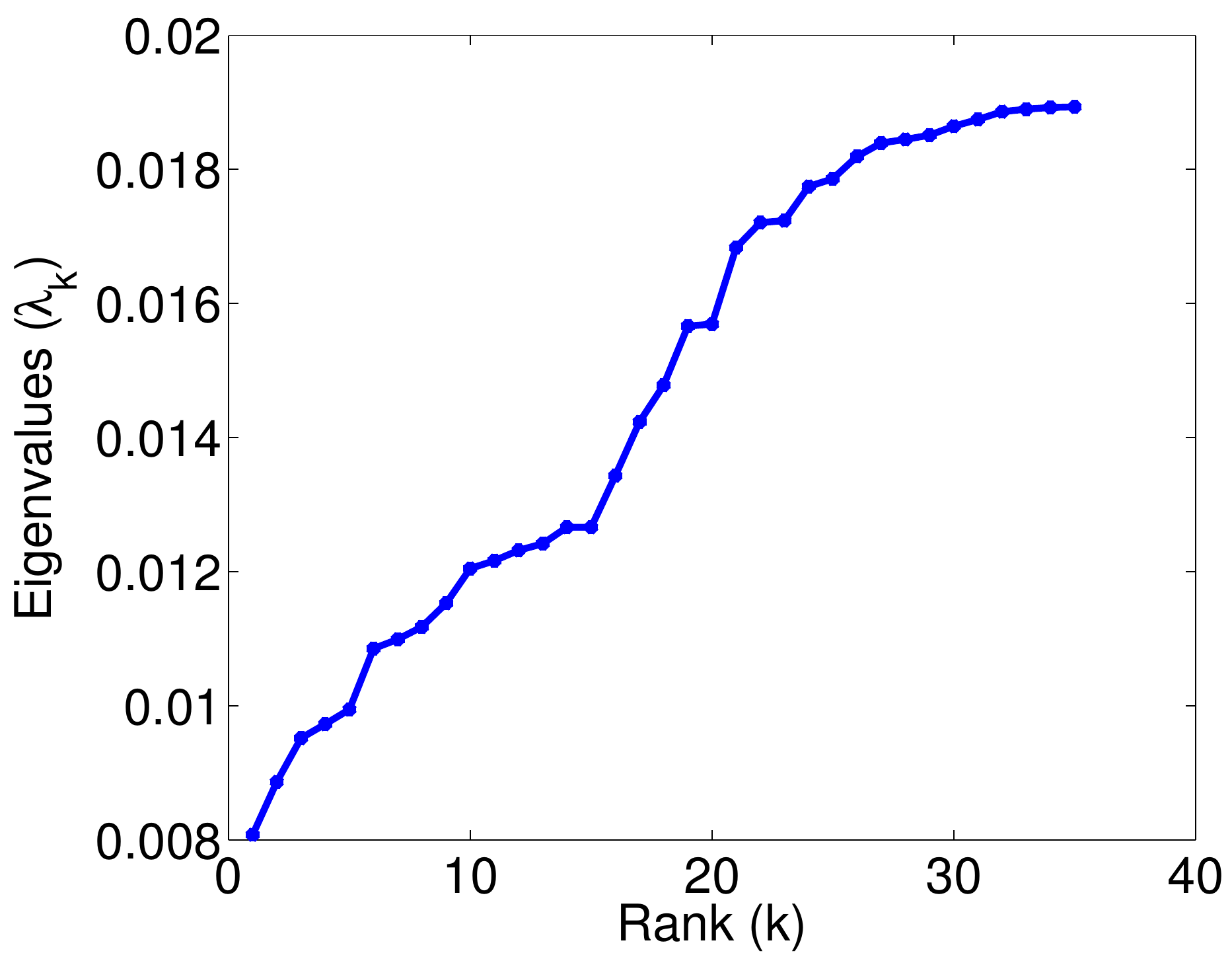}}
  \subfigure[Epinions]{
    \includegraphics[width=\wTwo]{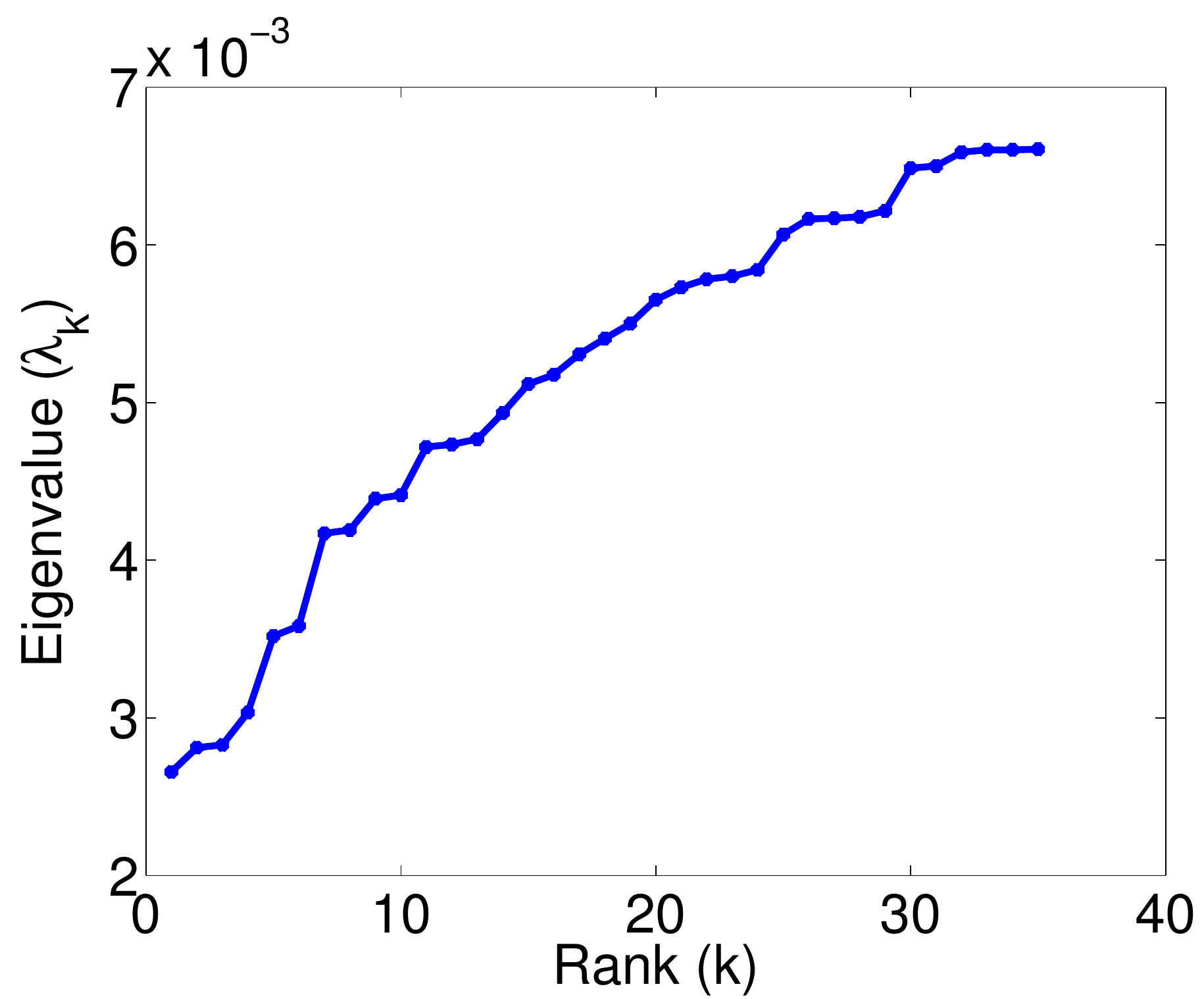}}
  \subfigure[Wikipedia elections]{
    \includegraphics[width=\wTwo]{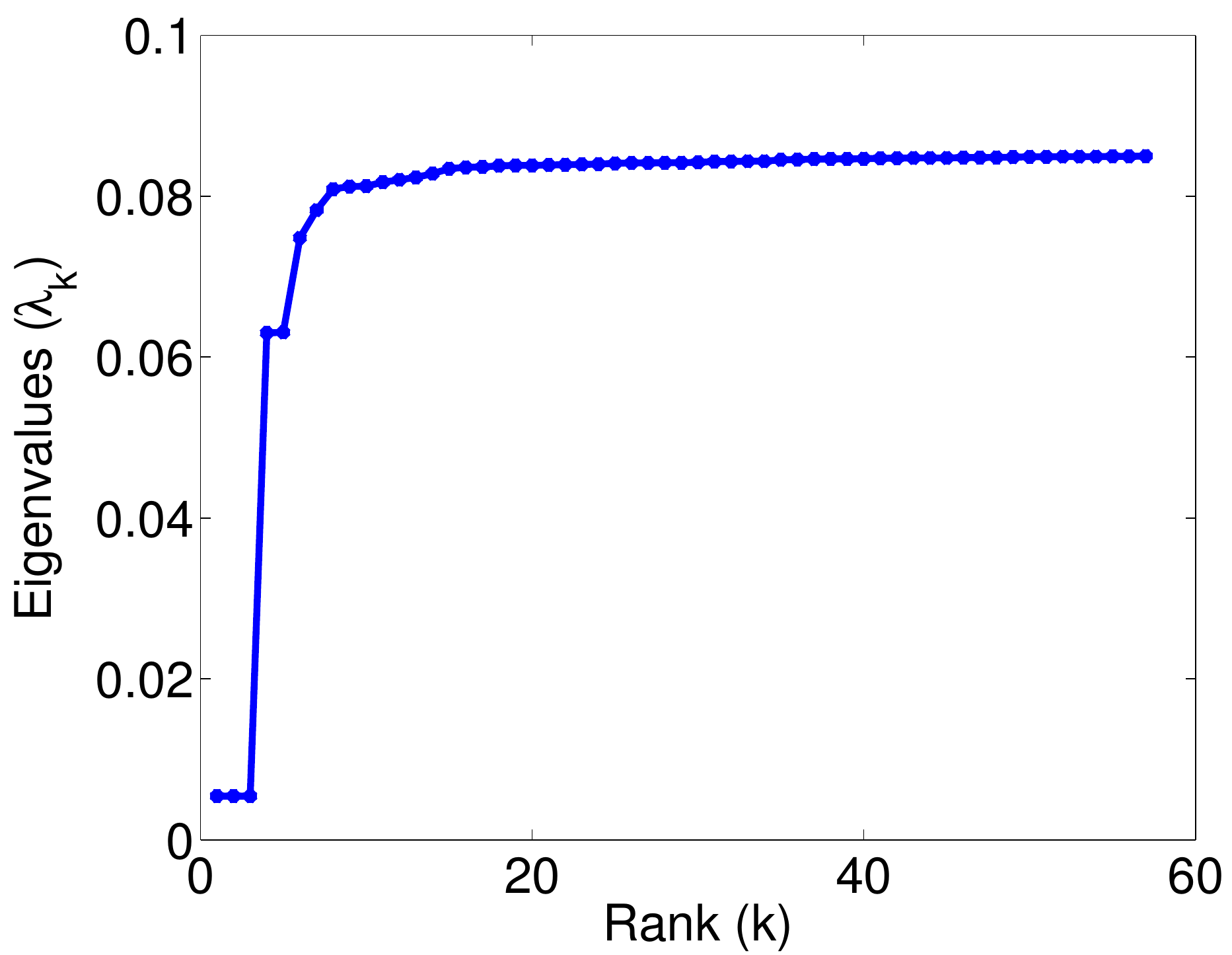}}
  \subfigure[Wikipedia conflicts]{
    \includegraphics[width=\wTwo]{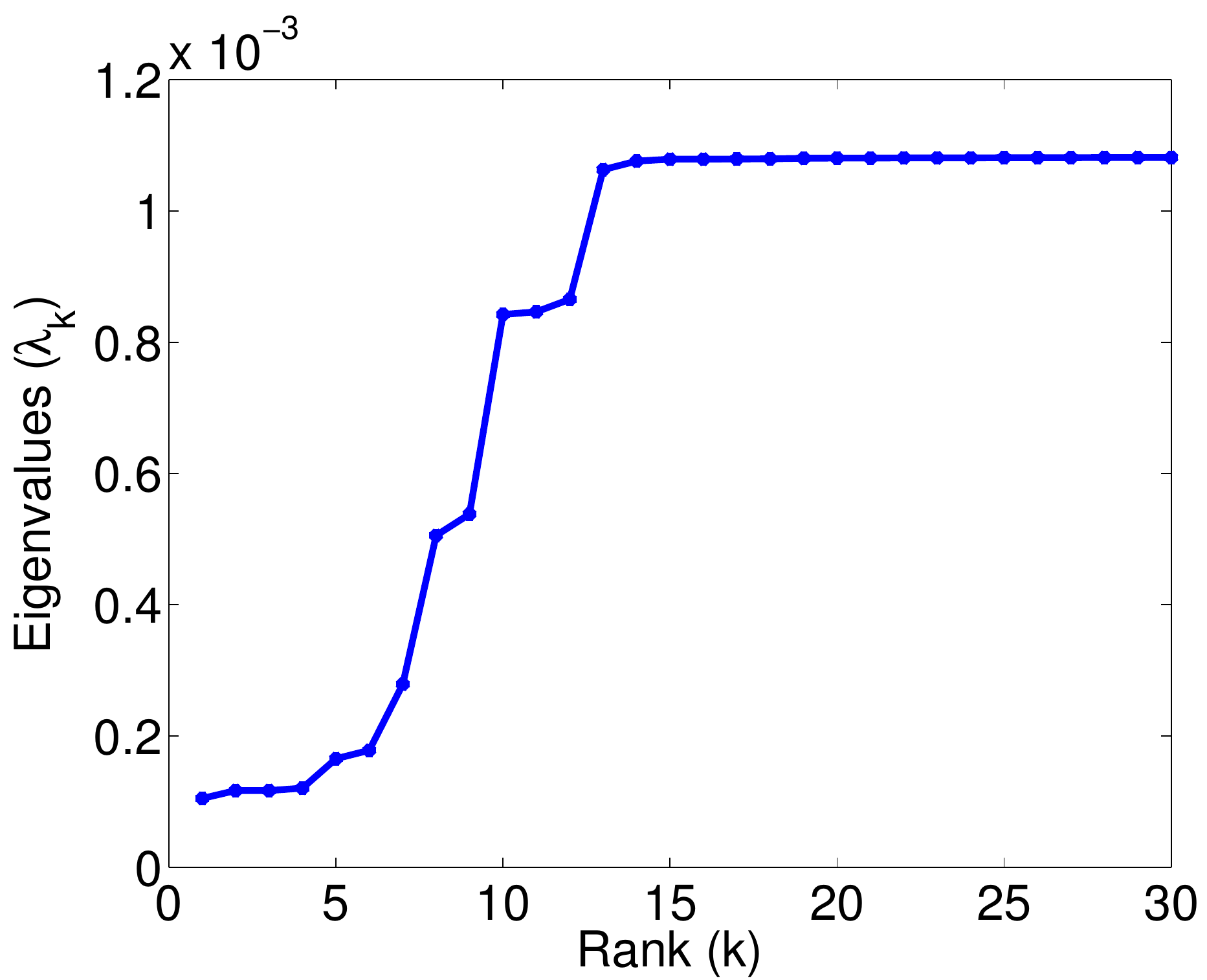}}
  \caption{
    The Laplacian spectra of three signed networks.  
    These plots show the eigenvalues $\lambda_1 \leq \lambda_2 \leq
    \cdots$ of the Laplacian matrix $\mathbf L$. 
  }   
  \label{fig:spectra}
\end{figure}

The spectra of several large unipartite signed networks
are plotted in Figure~\ref{fig:spectra}.  
We can observe that in all cases, the smallest eigenvalue is larger than
zero, implying, as expected, that these graphs are unbalanced. 

\section{Measuring Structural Balance 2:  \\ Algebraic Conflict}
\label{sec:conflict}
The smallest eigenvalue of the Laplacian $\mathbf L$ of a signed graph
is zero when the graph is balanced, and larger otherwise. 
We derive from this that the smallest Laplacian eigenvalue
characterizes the amount of conflict present in the graph.  We will call
this number the \emph{algebraic conflict} of the graph and denote
it~$\xi$.   

Let $G=(V,E,\sigma)$ be a connected signed graph with adjacency matrix $\mathbf A$,
degree matrix $\mathbf D$ and Laplacian $\mathbf L = \mathbf D
- \mathbf A$. 
Let $\lambda_1 \leq \lambda_2 \leq \cdots \leq \lambda_{|V|}$ be the
eigenvalues of $\mathbf L$.  
Because $\mathbf L$ is positive-semidefinite
(Theorem~\ref{theo:l-positive-semidefinite}), we have $\lambda_1 \geq 0$.
According to Theorem~\ref{theo:one}, $\lambda_1$ is zero exactly when $G$ is
balanced.  Therefore, the value $\lambda_1$ can be used as an
invariant of signed graphs that characterizes the conflict due to
unbalanced cycles, i.e., cycles with an odd number of negative edges.  
We will call $\xi = \lambda_1$ the \emph{algebraic conflict} of the network. 
The number $\xi$ is discussed in~\cite{b351}
and~\cite{kunegis:signed-kernels}, without being given a specific name.   

The algebraic conflict $\xi$ for our signed network datasets is
compared in 
Table~\ref{tab:smallest}. 
All these large networks are unbalanced, and we
can for instance observe that the social networks of the Slashdot Zoo and
Epinions are more balanced than the election network of Wikipedia. 

\begin{table}
  \centering
  \caption{
    The algebraic conflict $\xi$ for several signed unipartite networks.
    Smaller values indicate a more balanced network; larger values
    indicate more conflict.
  }
  % Values are taken from dat/statistic_full.conflict.$NETWORK
  \begin{tabular}{ l l }
    \toprule
    \textbf{Network} &  $\xi$ \\ 
    \midrule
    Slashdot Zoo & 0.008077 \\
    Epinions & 0.002657 \\
    Wikipedia elections  & 0.005437 \\
    Wikipedia conflicts & 0.0001050 \\
    \midrule
    Highland tribes & 0.7775 \\
    \bottomrule
  \end{tabular}
  \label{tab:smallest}
\end{table}

Figure~\ref{fig:scatter} plots the algebraic conflict
of the signed networks against the relative signed clustering coefficient
The number of signed
datasets is small, and thus we cannot make out a correlation between the two
measures, although the data is consistent with a negative between the
two measures, as expected. 
%% , indicating that the algebraic conflict is a meaningful
%% measure of conflict and does not have to be normalized.  
%% This conclusion
%% is however not very strong due to the fact that only few signed
%% networks were available to us.

\begin{figure}
  \centering
  \includegraphics[width=\wTwo]{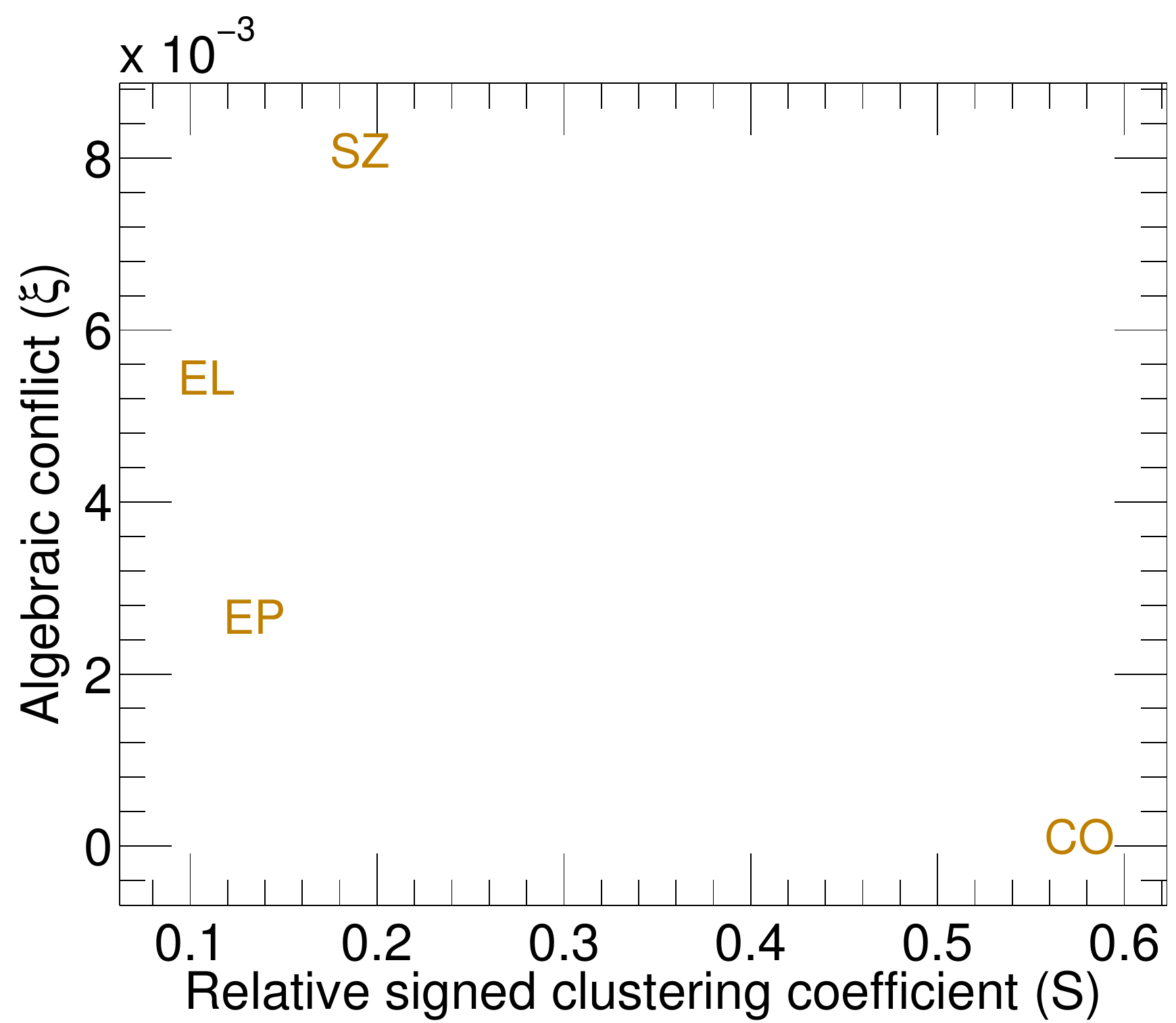}
  \caption{
    Scatter plot of the two measures of balance and conflict for the
    four signed social networks: The relative signed clustering
    coefficient $S$ and the algebraic conflict $\xi$. 
    (\textsf{SZ}: Slashdot Zoo, \textsf{EP}: Epinions, \textsf{EL}:
    Wikipedia elections, \textsf{CO}: Wikipedia conflict)
    %% The algebraic conflict $\xi$, i.e., the smallest eigenvalues of the matrix
    %% $\mathbf L = \mathbf D - \mathbf A$ in
    %% signed networks.  This plot shows the algebraic conflict and the
    %% size of the signed networks.
  }
  \label{fig:scatter}
\end{figure}

\paragraph{Monotonicity}
From the definition of the algebraic conflict $\xi$, we can derive a
simple theorem stating that adding an edge of any weight to a signed
graph can only increase the algebraic conflict, not decrease it. 
\begin{theorem}
  Let $G = (V, E, \sigma)$ be a signed graph and $u,v\in V$ two vertices
  such that $\{u,v\} \notin E$, and $\xi$ the algebraic
  of $G$.
  Furthermore, let $G' = (V, E \cup \{u,v\}, \sigma')$ with
  $\sigma'(e)=\sigma(e)$ when $e\in E$ and $\sigma(\{u,v\})=\sigma$ otherwise be the graph $G$ to
  which an edge with sign $\sigma$ has been added. 
  Then, let $\xi'$ be
  the algebraic conflict of
  $G'$. Then, $\xi \leq \xi'$. 
\end{theorem}
\begin{proof}
  We make use of a theorem stated for instance in~\cite[p.~97]{b663}.  This
  theorem states that when adding a 
  positive-semidefinite matrix $\mathbf E$ of rank one to a given
  symmetric matrix $\mathbf X$ 
  with eigenvalues $\lambda_1 \leq \lambda_2 \leq \cdots \leq
  \lambda_n$, the new matrix $\mathbf X' = \mathbf X + \mathbf E$ has
  eigenvalues $\lambda'_1 \leq \lambda'_2 \leq \cdots \leq
  \lambda'_n$ which interlace the eigenvalues of $\mathbf X$:
  \begin{eqnarray*}
    \lambda_1^{\phantom '} \leq \lambda'_1 \leq \lambda_2^{\phantom '} \leq \lambda'_2 
    \leq \cdots
    \leq \lambda_n^{\phantom '} \leq \lambda'_n
  \end{eqnarray*}

  The Laplacian $\mathbf L'$ of $G'$ can be written as $\mathbf L' =
  \mathbf L + \mathbf E$, where $\mathbf E \in \mathbb R^{|V| \times
    |V|}$ is the matrix defined by  
  $\mathbf E_{uu}=\mathbf E_{vv} = 1$ and $\mathbf E_{uv}=\mathbf E_{vu}
  = -\sigma$, and $\mathbf E_{uv} = 0 $ for all other entries. 
  Then let $\mathbf e \in \mathbb R^{|V|}$ be the vector defined by
  $\mathbf e_u = 1$, $\mathbf e_v = -\sigma$
  and $\mathbf e_w = 0$ for all other entries.  We have $\mathbf E =
  \mathbf e \mathbf e^{\mathrm T}$, and therefore $\mathbf E$ is
  positive-semidefinite.  

  Now, due to the interlacing theorem mentioned above, adding a
  positive-semidefinite matrix to a given symmetric matrix can only
  increase each eigenvalue, but not decrease it.  Therefore,
  $\lambda_1
  \leq \lambda'_1$, and thus $\xi \leq \xi'$. 
\end{proof}
We have thus proved that adding an edge of any sign to a signed network
can only increase the algebraic conflict, not decrease it. 
It also follows that removing an edge of any sign
from a signed network can decrease the algebraic conflict or leave it
unchanged, but not increase it. 

\section{Maximizing Structural Balance:  \\ Signed Spectral Clustering}
\label{sec:clustering}
One of the main application areas of the graph Laplacian are clustering problems.
In spectral clustering, the eigenvectors of matrices associated with a
graph are used to partition the vertices of the graph into
well-connected groups.  In this section, we show that in a signed graph,
the spectral clustering problem corresponds to 
finding clusters of vertices connected by positive edges, but not
connected by negative edges. 

Spectral clustering algorithms are usually derived by formulating a
minimum cut problem which is then
relaxed~\cite{b246,b304,b148,b453,b452}.  The choice of the cut 
function results in different spectral clustering algorithms.  In all
cases, the vertices of a given graph are mapped into the space spanned
by the eigenvectors of a matrix associated with the graph. 

In this section we derive a signed extension of the ratio cut, which
leads to clustering 
with the signed Laplacian $\mathbf L$.  
We restrict our proofs to
the case of clustering vertices into two groups;  higher-order
clusterings can be derived analogously.

\subsection{Unsigned Graphs}
We first review the derivation of the ratio cut in unsigned graphs. 
Let $G = (V,E)$ be an unsigned graph with adjacency matrix $\mathbf A$.
A cut of $G$ is a partition of 
the vertices $V$ into the nonempty sets $V_1$ and $V_2$, whose weight is
given by 
\begin{eqnarray*}
  \mathrm{Cut}(V_1,V_2) = |\{ \{u,v\} \in E \mid u \in V_1, v \in V_2 \}|. 
\end{eqnarray*}
The cut measures how
well two clusters are connected.  Since we want to find two distinct
groups of vertices, the cut must be minimized.  Minimizing
$\mathrm{Cut}(V_1,V_2)$ however leads in most cases to solutions separating
very few vertices from the rest of the graph.  Therefore, the cut is
usually divided by the size of the clusters, giving the ratio cut:
\begin{eqnarray*}
  \mathrm{RatioCut}(V_1,V_2) = 
  \left( \frac 1 {|V_1|} + \frac 1 {|V_2|} \right)
  \mathrm{Cut}(V_1,V_2) 
\end{eqnarray*}
To get a clustering, we then solve the following optimization problem: 
\begin{eqnarray*}
  \min_{V_1 \subset V} \quad \mathrm{RatioCut} (V_1, V \setminus V_1)
\end{eqnarray*}
Let $V_2 = V \setminus V_1$. Then
this problem can be solved by expressing it in terms of the 
characteristic vector $\mathbf x\in \mathbb R^{|V|}$ of $V_1$ defined by:
\begin{eqnarray}
  \mathbf x_u = \left\{ \begin{array}{ll}  +\sqrt{|V_2|/|V_1|} & \textnormal{ if }
    u \in V_1 \\
    -\sqrt{|V_1|/|V_2|} & \textnormal{ if } u \in V_2 \end{array} \right. 
  \label{eq:unsigned-clustering} 
\end{eqnarray}
We observe that $\mathbf x \mathbf L \mathbf x^{\mathrm T} = 2 |V| \cdot
\mathrm{RatioCut}(V_1,V_2)$, and that $\sum_u
\mathbf x_u = 0$, i.e., $\mathbf x$ is orthogonal to the constant
vector.  Denoting by 
$\mathcal X$ the vectors $\mathbf x$ of the form given in
Equation~(\ref{eq:unsigned-clustering}) we 
have 
\begin{eqnarray}
  \min_{\mathbf x \in \mathbb R^{|V|}} & \quad & \mathbf x \mathbf L \mathbf x^{\mathrm T} \\
  \mathrm{s.t.} & \quad & \mathbf x \cdot \mathbf 1 = 0, \mathbf x \in
  \mathcal X \nonumber
\end{eqnarray}
This can be relaxed by removing the constraint $\mathbf x \in \mathcal X$,
giving as solution the eigenvector of $\mathbf L$ having the smallest nonzero
eigenvalue~\cite{b304}.  

\subsection{Signed Graphs}
We now give a derivation of the ratio cut for signed graphs.  
Let $G=(V,E,\sigma)$ be a signed graph with adjacency matrix $\mathbf A$.  
We write $\mathbf A^\oplus$ and $\mathbf A^\ominus$ for
the adjacency matrices containing only the positive and
negative edges.  In 
other words, $\mathbf A^\oplus_{uv} = \max(0, \mathbf A_{uv})$, $\mathbf A^\ominus_{uv} =
\max(0,-\mathbf A_{uv})$ and $\mathbf A = \mathbf A^\oplus - \mathbf A^\ominus$. 

For convenience we define positive and negative cuts that only count
positive and negative edges respectively:
\begin{eqnarray*}
  \mathrm{Cut}^\oplus(V_1,V_2) &=& \sum_{u \in V_1, v \in V_2} \mathbf A^\oplus_{uv} \\
  \mathrm{Cut}^\ominus(V_1,V_2) &=& \sum_{u \in V_1, v \in V_2} \mathbf A^\ominus_{uv}
\end{eqnarray*}
In these definitions, we allow $V_1$ and $V_2$ to be overlapping. 
For a vector $\mathbf x \in \mathbb{R}^{|V|}$, we consider the bilinear
form $\mathbf x^{\mathrm T}
\mathbf L\mathbf x$.  As shown in Equation~(\ref{eq:bilinear}), this can be written
in the following way:
\begin{eqnarray*}
  \mathbf x^{\mathrm T} \mathbf L \mathbf x = \sum_{\{u,v\}\in E}
  (\mathbf x_u - \sigma(\{u,v\}) \mathbf x_v)^2
\end{eqnarray*}
For a given partition $V = V_1 \cup V_2$, let $\mathbf x \in \mathbb
R^{|V|}$ be the following vector:
\begin{eqnarray}
  \mathbf x_u = \left\{ \begin{array}{ll} +\frac 1 2 \left( \sqrt
    \frac{|V_1|}{|V_2|} + \sqrt 
    \frac{|V_2|}{|V_1|} \right) & \textnormal{ if } u \in V_1 \\
    -\frac 1 2 \left( \sqrt \frac{|V_1|}{|V_2|} + \sqrt
    \frac{|V_2|}{|V_1|} \right) & \textnormal{ if } u \in V_2 
  \end{array} \right. 
  \label{eq:signed-clustering} 
\end{eqnarray}
The corresponding bilinear form then becomes: 
\begin{eqnarray*}
  \mathbf x^{\mathrm T}\mathbf L\mathbf x &=& \sum_{\{u,v\}\in E}
  \left(\mathbf x_u -
  \sigma(\{u,v\})\mathbf x_v\right)^2 \\
  &=& |V| 
  \left(\frac 1 {|V_1|} +\frac 1 {|V_2|}\right)
  \left(2 \cdot \mathrm{Cut}^\oplus(V_1,V_2) + \mathrm{Cut}^\ominus(V_1,V_1)
  + \mathrm{Cut}^\ominus(V_2,V_2) \right) 
\end{eqnarray*}
This leads us to define the following signed cut of $(V_1,V_2)$:
\begin{eqnarray*}
  \mathrm{SignedCut}(V_1,V_2) &=&  \mathrm{Cut}^\oplus(V_1,V_2) +
  \frac 12 \left(\mathrm{Cut}^\ominus(V_1,V_1)  + \mathrm{Cut}^\ominus(V_2,V_2) \right)
\end{eqnarray*}
and to define the signed ratio cut as follows:
\begin{eqnarray*}
  \mathrm{SignedRatioCut}(V_1,V_2) = 
  \left(\frac 1 {|V_1|} +\frac 1 {|V_2|}\right) 
  \mathrm{SignedCut}(V_1,V_2)
\end{eqnarray*}
Therefore, the following minimization problem solves the signed
clustering problem:
\begin{eqnarray*}
  \min_{V_1 \subset V} \quad \mathrm{SignedRatioCut}(V_1, V \setminus V_1) 
\end{eqnarray*}
We can now express this minimization problem using the signed Laplacian,
where $\mathcal X$ denotes the set of vectors of the form given 
in Equation~(\ref{eq:signed-clustering}): 
\begin{eqnarray*}
  \min_{\mathbf x \in \mathbb R^{|V|}} && \quad \mathbf x\mathbf L
  \mathbf x^{\mathrm T} \\
  \mathrm{s.t.} && \quad \mathbf x \in \mathcal X 
\end{eqnarray*}
Note that we lose the orthogonality of $\mathbf x$ to the constant vector.  This
can be explained by the fact that if $G$ contains negative edges, the
smallest eigenvector can always be used for clustering:  If $G$ is
balanced, the smallest eigenvalue is zero and its eigenvector equals
$(\pm 1)$ and gives the two clusters separated by negative edges.  If $G$
is unbalanced, then the smallest eigenvalue of $\mathbf L$ is larger
than zero by Theorem~\ref{theo:one}, and the constant vector is not an
eigenvalue. 

The signed cut $\mathrm{SignedCut}(V_1,V_2)$ counts the number of positive edges that
connect the two groups $V_1$ and $V_2$, and the number of negative edges
that remain in each of these groups.  Thus, minimizing the signed cut
leads to clusterings where two groups are connected by few positive
edges and contain few negative edges inside each group. 
This signed ratio cut generalizes the ratio cut of unsigned graphs and
justifies the use of the signed Laplacian $\mathbf L$ and its particular
definition for spectral clustering of signed graphs. 

\subsection{Signed Clustering using Other Matrices}
When instead of normalizing with the number of vertices $|V_1|$ we
normalize with the number of edges $\mathrm{vol}(V_1)$, the result is a
spectral clustering algorithm based on the eigenvectors of
$\mathbf D^{-1}\mathbf A$ introduced by Shi and Malik~\cite{b452}.  The
cuts normalized 
by $\mathrm{vol}(V_1)$ are called normalized cuts. 
In the signed case, the eigenvectors of $\mathbf D^{-1}\mathbf A$ lead to the
signed normalized cut:
\begin{eqnarray*}
  \mathrm{SignedNormalizedCut}(V_1,V_2) &=&
  \left(\frac 1 {\mathrm{vol}(V_1)}+\frac 1 {\mathrm{vol}(V_2)}\right) 
  \mathrm{SignedCut}(V_1,V_2)
\end{eqnarray*}

A similar derivation can be made for normalized cuts based on $\mathbf N
= \mathbf D^{-1/2}\mathbf A\mathbf D^{-1/2}$, generalizing the spectral
clustering method of Ng, Jordan and Weiss~\cite{b453}.  The dominant
eigenvector of the signed adjacency matrix $\mathbf A$ can also be used
for signed clustering~\cite{b661}.  As in the unsigned case, this method
is not suited for very sparse graphs, and does not have an
interpretation in terms of cuts.  

\paragraph{Example}
As an application of signed spectral clustering to real-world data, we
cluster the tribes in the Highland tribes network. 
The resulting graph contains cycles with an odd number of negative edges, and
therefore its signed Laplacian matrix is positive-definite.
We use the eigenvectors of the two smallest eigenvalues ($\lambda_1 = 1.04$ and
$\lambda_2 = 2.10$) to embed the graph into the plane.  The result is shown in
Figure~\ref{fig:gama}. 
We observe that indeed the positive (green) edges are short, and the
negative (red) edges are long.
Looking at only the positive edges, the drawing makes the two connected
components easy to see.
Looking at only the negative edges, we recognize that the tribal groups
can be clustered into three groups, with no negative edges inside 
any group. 
These three groups correspond indeed to a higher-order grouping in the
Gahuku--Gama society~\cite{b323}.  

\begin{figure}
  \centering
  \includegraphics[width=\wOnePointFive]{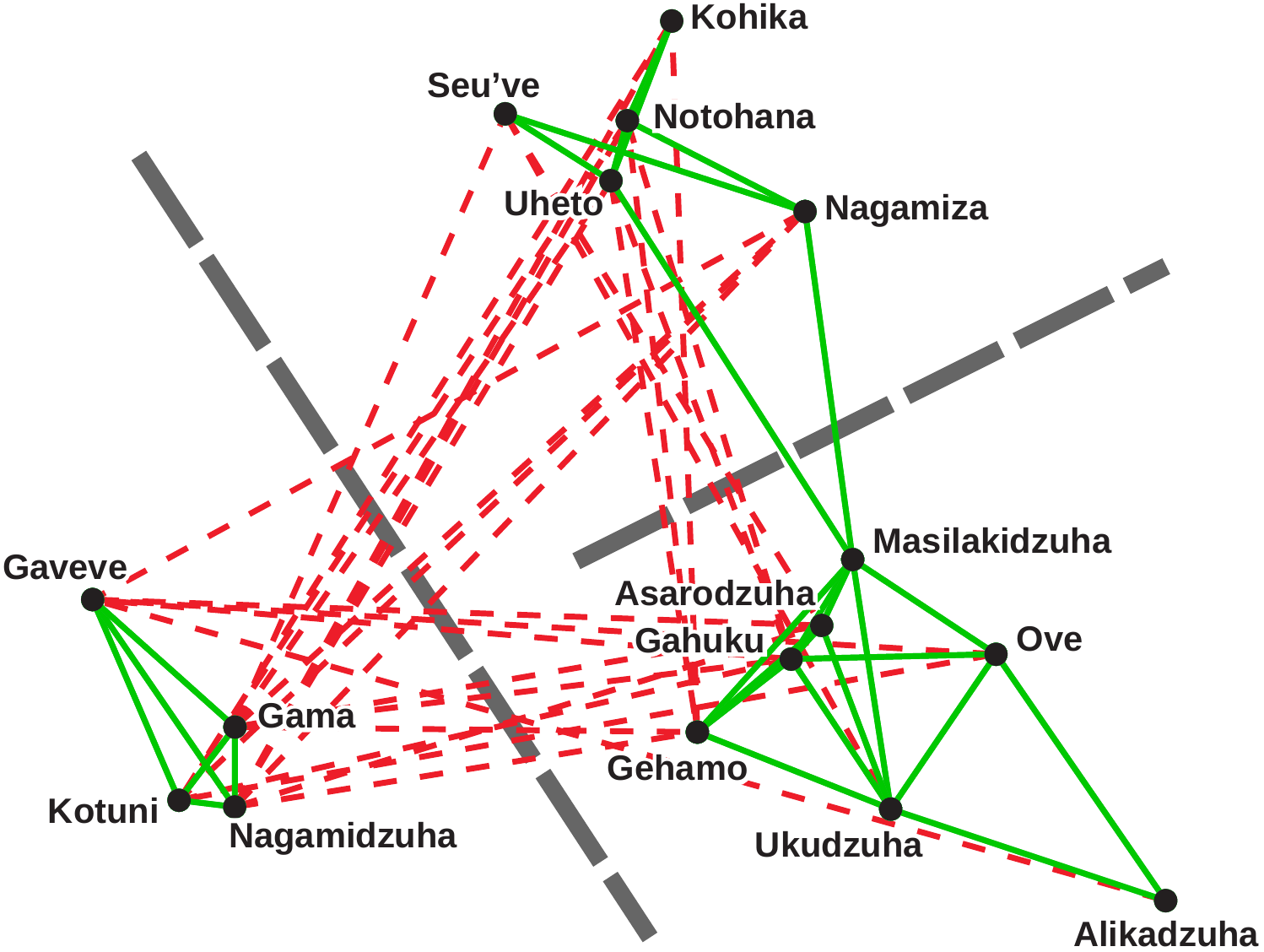} 
  \caption[*]{
    The tribal groups of the Eastern Central Highlands of New Guinea
    from the study of Read~\cite{b322} clustered using eigenvectors of
    the Laplacian matrix. 
    The three higher-order groups as described by Hage and
    Harary~\cite{b323} are linearly separable.  
  }
  \label{fig:gama}
\end{figure}

\section{Predicting Structural Balance: \\ Signed Resistance Distance}
\label{sec:prediction}
In the field of network analysis, one of the major applications consists
in predicting the state of an evolving network in the future. When
considering only the network structure, the corresponding learning
problem is the link prediction problem. In this section, we will 
show that a certain class of link prediction algorithms based on
algebraic graph theory are particularly suited to signed social
networks, since they fulfill three natural requirements that a link
prediction method should follow.  We will state the three
conditions, and then present two algebraic link prediction methods:  the
exponential of the adjacency matrix and the signed resistance
distance. We then finally evaluate the methods on the task of link
prediction. 

First however, let us give the correct terminology and define the link
prediction problem for unsigned and signed social networks. 
Although we state both problems in terms of social networks, both
problems can be extended to other networks. 

Actual social networks are not static graphs, but dynamic systems in
which nodes and edges are added and removed continuously.  The main type
of
change being the addition of edges, i.e., the appearance of a new
tie. Predicting such ties is a common task. For instance, social
networking sites try to predict who users are likely to already know in
order to give good friend recommendations. Let $G=(V,E)$ be an unsigned
social network.  The link prediction then consists of predicting new
edges in that network, a link prediction algorithm is thus a function
mapping a given network to edge predictions. In this work, we will
express link prediction functions algebraically as a map from the
$|V|\times |V|$ adjacency matrix of a network to another $|V|\times |V|$
matrix containing link prediction scores. The semantics of these scores
is that higher values denote a higher likelihood of link
formation. Apart from that, we do not put any other constraint on link
prediction scores. In particular, link prediction scores do not have to
be nonnegative, or restricted to the range $[0,1]$. 

In the case of signed social networks, the link prediction problem is
usually restricted to predicting positive edges. This is easily
motivated by the example of a social recommender system, which should
recommend friends and not enemies. Thus, the link prediction problem for
signed networks can be formalized in the same fashion as for unsigned
networks, by a function from the space of adjacency matrices (containing
positive and negative entries) to the space of score matrices. A link
prediction function $f$ for signed networks will thus be denoted as
follows:
\begin{eqnarray*}
  f : \{-1,0,+1\}^{|V|\times |V|} \rightarrow \mathbb R^{|V|\times |V|} 
\end{eqnarray*}

A note is in order about the related problem of \emph{link sign
  prediction}.  In the problem of link sign prediction, a signed
(social) network is given, along with a set of unweighted edges, and the
goal is the predict the sign of the edges~\cite{kunegis:slashdot-zoo,b555}.  
This problem is different from the link prediction problem in that for
each given edge, it is known that the edge is part of the network, and
only its sign must be predicted. By contrast, the link prediction
problem assumes no knowledge about the network and consists in finding
the positive edges. 

\paragraph{Requirements of a Link Prediction Function}
The structure of the link prediction problem implies two
requirements for a link prediction function, in relation with paths
connecting any two nodes. In addition, the presence of negative edges
implies a third requirement, in relation to the edge signs in paths
connecting two nodes.

Let $V$ be a fixed set of vertices, and $G_1=(V,E_1)$ and $G_2=(V,E_2)$
two unsigned networks with the same vertex sets. Let $u,v \in V$ be two
vertices and $f$ a link prediction function.  Then, compare the set set
of paths connecting the vertices $u$ and $v$, both in $G_1$ and in
$G_2$. Two requirements should be fulfilled by $f$:
\begin{itemize}
  \item \textbf{Path counts}:  If more paths between $u$ and $v$ are
    present in $G_1$ than $G_2$, than $f$ should return a higher score
    for the pair $(u,v)$ in $G_1$ than in $G_2$. 
  \item \textbf{Path lengths}: If paths between $u$ and $v$ are longer in
    $G_1$ than in $G_2$, then $f$ should return a lower score for the
    pair $(u,v)$ in $G_1$ than in $G_2$. 
\end{itemize}
In addition, the following requirement can be formulated for signed
networks. In this requirement, we will refer to a path
as positive when it contains an even number of negative edges and as negative
when it contains an odd number of negative edges. 
\begin{itemize}
  \item \textbf{Path signs}:  If paths between $u$ and $v$ are more
    often positive in $G_1$ than in $G_2$, 
    than $f$ should return a higher score
    for the pair $(u,v)$ in $G_1$ than in $G_2$. 
\end{itemize}

These three requirements are fulfilled by several link prediction
functions, of which we review one and introduce another in the
following.

\subsection{Signed Matrix Exponential}
Let $G=(V,E,\sigma)$ be a signed network with adjacency matrix $\mathbf
A$. Its exponential is then defined as
\begin{eqnarray*}
  e^{\alpha \mathbf A} &=& \mathbf I + \mathbf A + \frac 12 \mathbf A^2 +
  \frac 16 \mathbf A^3 + \cdots
\end{eqnarray*}
This exponential with the parameter $\alpha > 0$ is a suitable link
prediction function for signed networks as it can be expressed as a sum
over all paths between any two nodes. Let $P_G(u,v,k)$ be the set of
paths of length $k$ in the graph $G$. In this definition, we allow a
path to cross a single vertex multiple times, and set the length of a
path as being the number of edges it contains.  Furthermore, let
\begin{eqnarray*}
  (v_0, v_1, \ldots, v_k) \in P_G(u,v,k)
\end{eqnarray*}
with $u=v_0$ and $v=v_k$. Then, any power of $\mathbf A$ can be
expressed as
\begin{eqnarray*}
  (\mathbf A^k)_{uv} &=& \sum_{(v_0, \ldots, v_k) \in P_G(u,v,k)}
  \prod_{i=1}^k \sigma(\{v_{i-1}, v_i\}). 
\end{eqnarray*}
In other words, the $k$\textsuperscript{th} power of the adjacency
matrix equals a sum over all paths of length $k$, weighted by the
product of their edge signs. This leads to the following expression for
the matrix exponential:
\begin{eqnarray*}
  (e^{\alpha \mathbf A})_{uv} &=& \sum_{k=0}^\infty \frac {\alpha^k} {k!} \sum_{(v_0, \ldots, v_k) \in P_G(u,v,k)}
  \prod_{i=1}^k \sigma(\{v_{i-1}, v_i\}). 
\end{eqnarray*}
In other words, the matrix exponential is a sum over all paths between
any two nodes, weighted by the function $\alpha^k/k!$ of their path
length. This implies that the matrix exponential is a suitable link sign
prediction function for signed networks, since it fulfills all three
requirements:
\begin{itemize}
  \item \textbf{Path counts}:  The exponential function is a sum over
    paths and thus counts paths.
  \item \textbf{Path lengths}:  The function $\alpha^k/k!$ is decreasing
    in $k$, for suitably small values of $\alpha$. 
  \item \textbf{Path signs}:  Signs are taken into account by
    multiplication. 
\end{itemize}

Thus, the exponential of the adjacency matrix is a link prediction
function for signed networks. 

Other, similar functions can be constructed, for instance
the function $(\mathbf I - \alpha \mathbf A)^{-1}$ is known as the
Neumann kernel, 
in which $\alpha$ is chosen such that $\alpha^{-1} >
|\lambda_1|$, $|\lambda_1|$ being $\mathbf A$'s largest absolute eigenvalue,
or equivalently the graph's spectral norm~\cite{b263}.

Both the matrix exponential and the Neumann kernel can be applied to the
normalized adjacency matrix $\mathbf N = \mathbf D^{-1/2} \mathbf A
\mathbf D^{-1/2}$, in which each edge $\{u,v\}$ is weighted by
$\sqrt{d(u)d(v)}$, i.e., the geometric mean of the degrees of $u$ and
$v$. The rationale behind this normalization is to count a connection as
less important if it is one of many that attaches to a node. 

\subsection{Signed Resistance Distance}
The resistance distance is a metric defined on vertices of a graph
inspired from electrical resistance networks.  When an electrical current is applied
to an electrical network of resistors, the whole network acts as a single
resistor whose resistance is a function of the individual resistances.  
In such
an electrical network, any two nodes of the network can be taken as the endpoint of the total
resistance, giving a function defined between every pair of nodes.
As shown in~\cite{b101}, this function is a metric, usually called the
\emph{resistance distance}. 

Intuitively, two nodes further apart are separated by a greater equivalent
resistance, while 
nodes closer to each other lead to a small resistance distance.  This
distance function has been used before to perform collaborative
filtering~\cite{b105,b15,b107}, and it fulfills the first two of our
assumptions, when actual edge weights are interpreted as inverse
resistances, i.e., conductances:
\begin{itemize}
  \item \textbf{Path counts}: Parallel resistances are inverse-additive,
    and parallel conductances are additive.
  \item \textbf{Path lengths}:  Resistances in series are additive and
    conductances in series inverse-additive. 
\end{itemize}
As the resistance distance by default only applies to
nonnegative values, previous works use it on
nonnegative data, such as unsigned social networks or document view counts.  
In the presence of signed edges, the resistance distance can be extended by
the following formalism, which fulfills the third requirement on path signs. 
A positive electrical resistance indicates that the potentials
of two connected nodes will tend to each other:  The smaller the
resistance, the more both potentials approach each other.  
Therefore, a positive edge can be represented as a unit resistor. 
If an edge is negative, we can interpret the connection as
consisting of a unit resistor in series
with an \emph{inverting amplifier} that guarantees its ends to have
opposite voltage, as depicted in Figure~\ref{fig:elec}.  In other
words, two nodes connected by a negative 
edge will tend to opposite voltages. 

\begin{figure}
  \centering
  \includegraphics[width=\wOnePointFive]{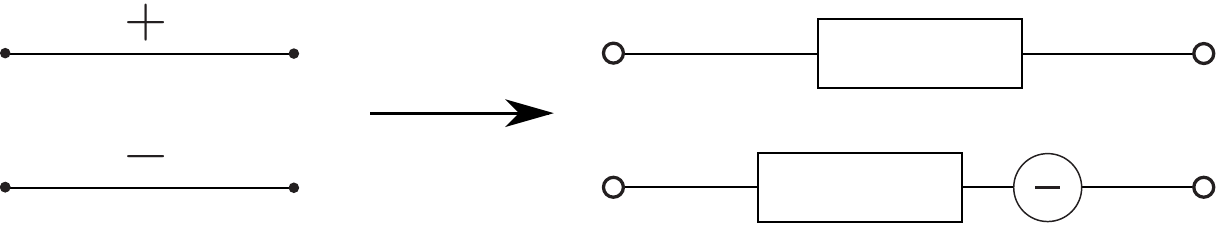}
  \caption{ 
    Interpretation of positive and negative edges as electrical
    components. 
    An edge with a negative weight is interpreted as a positive
    resistor in series with an inverting component (shown as~$\circleddash$). 
  }
  \label{fig:elec}
\end{figure}

Thus, a positive edge can be modeled by a unit resistor and a negative
edge can be modeled by a unit resistor in series with a (hypothetical)
electrical component that assures its ends have opposite electrical
potential.
Note that the absence of an edge is modeled by the absence of a
resistor, which is equivalent to a resistor with infinite
resistance. Thus, actual edge weights and scores correspond not to
resistances, but to inverse resistance, i.e., conductances. 

We now establish a closed-form
expression giving the resistance distance between all node pairs based
on~\cite{b101}. 

\paragraph{Definitions}  
The following notation is used.
\begin{itemize}
\item $\mathbf J_{uv}$ is the current flowing through the oriented edge $(u,v)$.
  $\mathbf J$ is skew-symmetric: $\mathbf J_{uv} = - \mathbf J_{vu}$. 
\item $\mathbf{v}_u$ is the electric potential at node $u$.  Potentials are defined up to
  an additive constant.
\item $\mathbf R_{uv}$ is the resistance value of edge $(u, v)$.  Note
  that $\mathbf R_{uv}= \mathbf R_{vu}$. 
\end{itemize}
In electrical networks, the current entering a node must be equal to the
current leaving that node.  This relation is known as \emph{Kirchhoff's law}, and can
be expressed as $\sum_{v\sim u} \mathbf J_{uv}  =  0$ for all $u \in V$. 
We assume that a current $j$ will be flowing through the network from
vertex $a$ to vertex $b$, and therefore we have
\begin{eqnarray*}
  \sum_{(v, a)} \mathbf J_{av} &=& j, \\
  \sum_{(v, b)} \mathbf J_{bv} &=& -j.
\end{eqnarray*}
Using the identity matrix $\mathbf I$, we express these
relations as
\begin{eqnarray}
  \sum_{(v, u)} \mathbf J_{uv} = j (\mathbf I_{ua} - \mathbf I_{ub}) \label{eq:law}
\end{eqnarray}
The relation between currents and potentials is given by Ohm's law: $\mathbf{v}_u
- \mathbf{v}_v = \mathbf R_{uv} \mathbf J_{uv}$ for all edges $(u, v)$. 

We will now show that the equivalent resistance $\mathbf{\bar R}_{ab}$ between $a$ and
$b$ in the network can be expressed in terms of the graph Laplacian
$\mathbf L$ as
\begin{eqnarray}
  \mathbf{\bar R}_{ab} &=& (\mathbf I_{a\bullet} - \mathbf I_{b\bullet}) \mathbf L^+ (\mathbf I_{a\bullet} -
  \mathbf I_{b\bullet})^{\mathrm T}, \label{eq:unsigned} \\
  &=& \mathbf L^+_{aa} + \mathbf L^+_{bb} - \mathbf L^+_{ab} - \mathbf
    L^+_{ba}, \nonumber
\end{eqnarray}
where $\mathbf L^+$ is the Moore--Penrose pseudoinverse of $\mathbf L$~\cite{b101}. 

The proof follows from recasting Equation~(\ref{eq:law}) as:
\begin{eqnarray*}
  \sum_{(v, u)} \frac{1}{\mathbf R_{uv}} (\mathbf v_u - \mathbf v_v) &= j
  (\mathbf I_{ua} - \mathbf I_{ub})  \nonumber
\end{eqnarray*}
Combining over all $u \in V$:
\begin{eqnarray*}
  \mathbf D \mathbf{v} - \mathbf A\mathbf{v} &=& j(\mathbf I_{a\bullet} - \mathbf I_{b\bullet}) \\
  \mathbf L\mathbf{v} &=& j(\mathbf I_{a\bullet} - \mathbf I_{b\bullet})
\end{eqnarray*}
Let $\mathbf L^+$ be the Moore--Penrose pseudoinverse of $\mathbf L$, then
because $\mathbf{v}$ is contained in the row space of $\mathbf L$~\cite{b101}, we have
$\mathbf L^+\mathbf L\mathbf{v} = \mathbf{v}$, and we get 
\begin{eqnarray*}
  \mathbf{v} &=& \mathbf L^+ j (\mathbf I_{a\bullet} - \mathbf I_{b\bullet}) 
\end{eqnarray*}
Which finally gives the equivalent resistance between $a$ and $b$ as
\begin{eqnarray*}
  \mathbf{\tilde{R}}_{ab} &=& (\mathbf v_a - \mathbf v_b) / j \\
  &=& (\mathbf I_{a\bullet} - \mathbf I_{b\bullet})^{\mathrm T} \mathbf{v} / j \\
  &=& (\mathbf I_{a\bullet} - \mathbf I_{b\bullet})^{\mathrm T} \mathbf L^+ (\mathbf I_{a\bullet} -
  \mathbf I_{b\bullet})
\end{eqnarray*}
A symmetry argument shows that $\mathbf{\tilde R}_{ab} = \mathbf{\tilde R}_{ba}$ as expected.
As shown in~\cite{b101}, $\mathbf{\tilde R}$ is a metric. 

The definition of the resistance distance can be extended to signed
networks in the following way. 
\begin{figure}
  \centerline{
  \xymatrix @R-20pt
  {
    (a) & \node \ar@{-}[r]^{r_1=+1} & \node \ar@{-}[r]^{r_2=-1} & \node & 
         r = r_1 + r_2 = 0 \\
	 \\
    (b) & \node \ar@{-}@/^/[rr]^{r_1=+1}  \ar@{-}@/_/[rr]^{r_2=-1} & & \node & r =
  \frac{r_1 r_2}{r_1 + r_2}  = -1 / 0\\
  }
  }
  \caption{
    Applying the sum rules to negative resistance values leads
    to contradictions.  
  }
  \label{fig:ex}
\end{figure}
Figure~\ref{fig:ex} shows two examples in which we allow negative
resistance values in Equation~(\ref{eq:unsigned}):  two
parallel edges, and two serial edges.  In these examples, we will use the sum
rules that hold for electrical resistances:  resistances in series add up and
conductances in parallel also add up. 

Therefore, the constructions of Figure~\ref{fig:ex}. would result in a
total resistance of zero for case (a), and an undefined total resistance
in case (b).  However, the graph of Figure~\ref{fig:ex}~(a) could result
from two users $a$ and $b$ having a positive and a negative correlation
with a third user $c$.  Intuitively, the resulting distance between $a$
and $b$ should take on a negative value.  In the graph of
Figure~\ref{fig:ex}~(b), the intuitive result would be $r = -1/2$.  What
we would like is for the sign and magnitude of the equivalent resistance
to be handled separately: The sum rules should hold for the
\emph{absolute values} of the resistance similarity values, while the
sign should obey a product rule.  These requirements are summarized in
Figure~\ref{fig:summ}.

\begin{figure}
  \centerline{
  \xymatrix @R-20pt
  {
    (a) & \node \ar@{-}[r]^{r_1=+1} & \node \ar@{-}[r]^{r_2=-1} & \node & 
         r = \mathrm{sgn}(r_1 r_2) (|r_1| + |r_2|) = -2 \\
	 \\
    (b) & \node \ar@{-}@/^/[rr]^{r_1=+1}  \ar@{-}@/_/[rr]^{r_2=-1} & & \node & r =
         \frac{r_1 r_2 }{|r_1| + |r_2|} =  -1/2\\
  }
  }
  \caption{
    Applying modified sum rules resolves the contradictions. 
  }
  \label{fig:summ}
\end{figure}
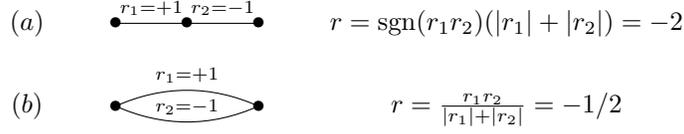
To achieve the serial sum equation proposed in Figure~\ref{fig:summ}, we
propose the following interpretation of a negative resistance:
\begin{itemize}
\item An edge carrying a negative resistance value acts like the corresponding
  positive resistance in series with a component that negates potentials.
\end{itemize}
A component that negates electric potential cannot exist in physical 
electrical networks, because it violates an invariant of electrical circuit:
The invariant stating that potentials are only defined up to an additive
constant.  However, as we will see below, the potential inversion gets canceled
out in the calculations, yielding results independent of any additive constant.
For this reason, we will talk of negative resistances, but avoid the term
resistor in this context. 

Before giving a closed-form expression for the signed resistance distance, we
provide three intuitive examples validating our definition in
Figure~\ref{fig:three_examples}. 

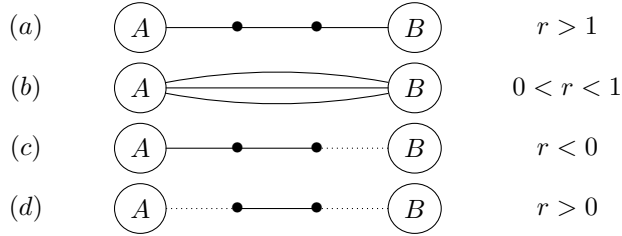
\begin{figure}
  \centerline{
  \xymatrix @R-20pt 
  {
    (a) &\circlenode{A}  \ar@{-}[r] & \node \ar@{-}[r] & \node \ar@{-}[r] &
  \circlenode{B} & r > 1 \\
    (b) & \circlenode{A} \ar@{-}@/^/[rrr] \ar@{-}[rrr] \ar@{-}@/_/[rrr] & & &
  \circlenode{B} & 0 < r < 1 \\
    (c) & \circlenode{A}  \ar@{-}[r] & \node \ar@{-}[r] & \node \ar@{.}[r] &
  \circlenode{B} & r < 0 \\
    (d) & \circlenode{A}  \ar@{.}[r] & \node \ar@{-}[r] & \node \ar@{.}[r] &
  \circlenode{B} & r > 0
  }
  }
  \caption{
    Example configurations of signed resistance values.  The total
    resistance is to be calculated between the nodes A and B.  All edges
    have unit absolute resistance.  Edges with negative resistance
    values are shown as dotted lines.  For each case, we formulate a
    condition that should hold for any signed resistance distance.
  }
  \label{fig:three_examples}
\end{figure}

\begin{itemize}
\item Example (a) shows that, as a path of resistances in series gets longer, the
  resulting resistance increases.  This conditions applies to the regular
  resistance distance as well as to the signed resistance distance.  In this
  case, the total resistance should be higher than one. 
\item Example (b) shows that a higher number of parallel
  resistances decreases the resulting resistance value.  Again, this is true for
  both types of resistances.  In this example, the total resistance should be
  less than one. 
\item Examples (c) and (d) show that in a path of signed resistances, the total
  resistance has the sign of the product of individual resistances.  This
  condition is particular to the signed resistance distance.  
\end{itemize}

We will now show how Kirchhoff's law has to be adapted to support our
definition of negative resistances.  We adapt Equation~(\ref{eq:law}) by applying
the absolute value to the resistance weight. 
\begin{eqnarray*}
  \sum_{(v, u)} \frac{1}{|\mathbf R_{uv}|} (\mathbf v_u -
  \mathrm{sgn}({\mathbf R_{uv}})
  \mathbf v_v) = 0
\end{eqnarray*}
where $\mathrm{sgn}(x)$ denotes the sign function. 
In terms of the matrices $\mathbf D$ and $\mathbf L$ we arrive at
\begin{eqnarray*}
  \mathbf D_{uu} &=& \sum_{(v, u)} |1/\mathbf R_{uv}|  \\
  \mathbf L &=& \mathbf D - \mathbf A \\
  \mathbf{\tilde {r}}_{ab} &=& (\mathbf I_{a\bullet} - \mathbf I_{b\bullet}) \mathbf L^+
  (\mathbf I_{a\bullet} - \mathbf I_{b\bullet})^{\mathrm T} \\
  &=& \mathbf L^+_{aa} +\mathbf L^+_{bb}-\mathbf L^+_{ab}-\mathbf L^+_{ba}. \nonumber
\end{eqnarray*}
The proof follows analogously to the proof for the regular resistance distance
by noting that $\mathbf{v}$ is again contained in the row space of $\mathbf L$. 
\begin{eqnarray*}
  \mathbf L^+ \mathbf L \mathbf{v} = \mathbf{v} 
\end{eqnarray*}
From which the result follows. 

As with the regular resistance distance, the signed resistance distance is
symmetric:  $\mathbf{\tilde R}_{ab} = \mathbf{\tilde R}_{ba}$. 

Due to a duality between electrical networks and random
walks~\cite{b19}, the resistance distance is also known as the
commute-time kernel, and its values can be interpreted as the average
time it takes a random walk to \emph{commute}, i.e., to go from a node $u$ to
another node $v$ and back to $u$ again. 

The matrix $\mathbf L^+$ will be called the resistance distance
kernel. Similarly, the matrix $\mathbf e^{-\alpha \mathbf L}$ is known
as the heat diffusion kernel, because it can be derived from a physical
process of heat diffusion. 
Both of these kernels can be normalized, i.e., they can be applied to
the normalized adjacency matrix $\mathbf N=\mathbf D^{-1/2} \mathbf A
\mathbf D^{-1/2}$, giving the normalized resistance distance kernel and
the normalized heat diffusion kernel.  
We note that the normalized heat diffusion kernel is equivalent to the
normalized exponential kernel~\cite{b155}. 

The degree matrix $\mathbf D$ of a signed graph is defined in this
article using $\mathbf D_{uu} = \sum_v |\mathbf A_{uv}|$ in the general case.  In some
contexts, an alternative degree matrix $\mathbf D_{\mathrm{alt}}$ is defined
without the absolute value:  
\begin{eqnarray*}
  (\mathbf D_{\mathrm{alt}})_{uu} = \sum_v \mathbf A_{uv}
\end{eqnarray*}
This leads to an alternative Laplacian matrix $\mathbf L_{\mathrm{alt}} =
\mathbf D_{\mathrm{alt}}-\mathbf A$ 
for signed graphs that is not positive-semidefinite.
This Laplacian is used in the context of knot theory~\cite{b365}, 
to draw graphs with negative edge weights~\cite{b357}, 
and to implement constrained clustering, i.e., clustering with \emph{must-link}
and \emph{must-not-link} edges~\cite{b558}. 
Since $\mathbf L_{\mathrm{alt}}$ is not positive-semidefinite in the general
case, it cannot be used as a kernel. 

Expressions of the form $(\sum_i |\mathbf w_i|)^{-1} \sum_i \mathbf w_i
\mathbf x_i$ appeared several times in the 
preceding sections.  These types of expressions represent a weighted mean of the
values $\mathbf x_i$, supporting negative values of the weights $\mathbf w_i$. 
These expressions have been used for some
time in the collaborative filtering literature without being connected
to the signed Laplacian, for instance in~\cite{b132}.

\subsection{Evaluation}
We compare the methods shown in Table~\ref{tab:methods} at the task of
link prediction in signed social networks. 

Evaluation is performed using the following methodology. Let
$G=(V,E,\sigma)$ be any of the signed networks, and let
\begin{eqnarray*}
  E = E_{\mathrm a} \cup E_{\mathrm b}
\end{eqnarray*}
be a partition of the edge set $E$ into a training set $E_{\mathrm a}$
and a test set $E_{\mathrm b}$. The training set is chosen to comprise
75\% of all edges. For the networks in which edge arrival times are
known (Epinions, Wikipedia elections, Wikipedia conflict), the split is
made in such a way that all edges in the training set $E_{\mathrm a}$
are older than the edges in the test set $E_{\mathrm b}$. 
Each link prediction method is then applied to the training network
\begin{eqnarray*}
  G_{\mathrm a} = (V, E_{\mathrm a}). 
\end{eqnarray*}
Let $E^+_{\mathrm b}$ denote the test edges with positive sign. 
Then, a zero test set $E_{\mathrm z}$ of edges not in the network at all
is generated, having the same size as $E^+_{\mathrm b}$. Then, the scores
of each link prediction algorithm are computed for all node pairs in
$E^+_{\mathrm b}$ and $E_{\mathrm z}$, and the accuracy of each link
prediction algorithm evaluated on $E^+_{\mathrm b}$ and $E_{\mathrm z}$
using the area under the curve (AUC) measure~\cite{b366}. 
The area under the curve is a number in the range $[0,1]$ which is larger
for better predictions, and admits a value of 0.5 for a random predictor. 
The parameters $\alpha$ of the various link prediction functions are
learned using the method described
in~\cite{kunegis:spectral-transformation}. 
The results of the experiments are shown in Table~\ref{tab:results}. 

\begin{table}
  \caption{
    The link prediction functions evaluated on the signed social network
    datasets. 
    Each method is a function of a specific characteristic graph matrix:
    $\mathbf A$, the adjacency matrix; $\mathbf N= \mathbf D^{-1/2}
    \mathbf A \mathbf D^{-1/2}$, the normalized adjacency matrix;
    $\mathbf L = \mathbf D - \mathbf A$, the Laplacian matrix; and
    $\mathbf Z = \mathbf I - \mathbf N = \mathbf D^{-1/2} \mathbf L
    \mathbf D^{-1/2}$, the normalized Laplacian matrix. 
  }
  \centering
  \scalebox{0.9}{
  \begin{tabular}{ll}
    \toprule
    \textbf{Name} & \textbf{Expression} \\
    \midrule
    Exponential (Exp) & $e^{\alpha \mathbf A}, 0 < \alpha$ \\
    Neumann kernel (Neu) & $(\mathbf I - \alpha \mathbf A)^{-1}, 0 < \alpha
    < \mathopen\parallel \mathbf A \mathclose\parallel_2^{-1}$ \\
    Normalized exponential (N-Exp) & $e^{\alpha \mathbf A}, 0 < \alpha$ \\
    Normalized Neumann kernel (N-Neu) & $(\mathbf I - \alpha \mathbf N)^{-1}, 0
    < \alpha < 1$ \\
    Resistance distance (Resi) & $\mathbf L^+$ \\
    Heat diffusion (Heat) & $e^{-\alpha \mathbf L}, 0 < \alpha$ \\
    Normalized resistance distance (N-Resi) & $\mathbf Z^+$ \\
    Normalized heat diffusion & \textit{Equivalent to Normalized exponential} \\
    \bottomrule
  \end{tabular}
  }
  \label{tab:methods}
\end{table}

\begin{table}
  \caption{
    The full evaluation results. The numbers are the area under the curve
    values (AUC); higher values denote better link prediction accuracy. 
    The best performing link prediction algorithm for each dataset is
    highlighted in bold. 
  }
  \scalebox{0.82}{
  \begin{tabular}{ l rrrrrrrr }
    \toprule
    \textbf{Network} & \textbf{Exp} & \textbf{Neu} & \textbf{N-Exp} & \textbf{N-Neu} 
    & \textbf{Resi} & \textbf{Heat} & \textbf{N-Resi} \\
    \midrule
    Slashdot Zoo        & \textbf{68.98\%} & 67.71\% & 64.87\% & 65.68\% & 61.64\% & 59.11\% & 65.71\% \\
    Epinions            & 75.04\% & 73.12\% & 78.38\% & 78.65\% & 63.26\% & 63.28\% & \textbf{78.82\%} \\
    Wikipedia elections & 57.08\% & 55.60\% & 60.30\% & \textbf{61.16\%} & 51.44\% & 50.60\% & 60.98\% \\
    Wikipedia conflicts & 85.57\% & 85.56\% & 85.03\% & 85.03\% & \textbf{87.02\%} & 85.95\% & 85.04\% \\
    \bottomrule
  \end{tabular}
  }
  \label{tab:results}
\end{table}

We observe that the best link prediction method depends on the
dataset. Each of the exponential, the normalized Neumann kernel,
the resistance distance kernel and the normalized resistance distance
kernel performs best for one or more datasets. 

\section{Conclusion}
\label{sec:conclusion}
We have reviewed network analysis methods for signed social networks~--
social networks that allow positive and negative edges.  A main theme we found
is that of structural balance, the statement that triangles in a signed
social network tend to be balanced, and on a larger scale the tendency
of a whole network to have a structure conforming to that assumption. We
showed how this can be measured in two different ways:  on the scale of triangles
by the signed clustering coefficient, and on the global scale by
the algebraic conflict, the smallest eigenvalue of the graph Laplacian. 
We also showed how structural balance can be exploited for graph
drawing, graph clustering, and finally for implementing social
recommenders, using signed link prediction algorithms. 

As structural balance can be seen as a form of multiplication rule
(illustrated by the phrase \emph{the enemy of my enemy is my friend}),
it is expected that algebraic methods are well-suited to analysing
signed social networks. Indeed, 
we identified functions of the adjacency matrix $\mathbf A$ and of the
Laplacian matrix $\mathbf L$, which model negative edges in a
natural way. 

In a more general sense, signed social networks can be understood as a
stepping stone to the more general topic of \emph{semantic networks}, in
which edges are labeled by arbitrary predicates. In such networks, the
combination of labels to give a new label, in analogy with the
multiplication rule of the signed edge weights $\{\pm 1\}$, cannot be
directly mapped by real numbers, and a general method for that case is
still an open problem in network theory. Certain subproblems have
however already be identified, for instance the usage of 
split-complex imaginary numbers to represent the \emph{like}
relationship~\cite{kunegis:split-complex-dating}.  

\section*{Acknowledgments}
We thank Andreas Lommatzsch, Christian Bauckhage, Stephan Schmidt,
Jürgen Lerner and Martin Mehlitz. 
The research leading to these results has received funding from the European Community's 
Seventh Frame Programme under grant agreement n\textsuperscript{o}~257859, 
ROBUST.

\bibliographystyle{abbrv}
\bibliography{ref,kunegis,konect} 

\end{document}